\documentclass[a4paper]{llncs}

\usepackage[usenames]{color}
\usepackage{fit_tr}

\usepackage{paralist}
\usepackage{comment}
\usepackage{amsmath}
\usepackage{amssymb}

\usepackage{tikz}
\usetikzlibrary{automata,arrows,patterns,shapes,snakes}

\usepackage[utf8]{inputenc}
\usepackage[USenglish]{babel}

\usepackage{pslatex}

\usepackage{sty/comments}
\usepackage{sty/mathematics}
\usepackage{sty/programs}

\usepackage{todonotes}
\usepackage{tabularx}
\usepackage{multirow}
\usepackage{booktabs}
\usepackage{arydshln}

\usepackage{listings}
\definecolor{dkgreen}{rgb}{0,0.6,0}
\definecolor{mauve}{rgb}{0.58,0,0.82}
\makeatletter
\newcommand{\leqnomode}{\tagsleft@true}
\newcommand{\reqnomode}{\tagsleft@false}
\makeatother

\title{Pointer Race Freedom
\thanks{This work was supported by the Czech Science Foundation, project 13-37876P, and by the German Science Foundation (DFG), project R2M2.}
}
\author{Fr\'ed\'eric Haziza\inst{1} \and Luk\'a\v s Hol\'ik\inst{2}\and Roland Meyer\inst{3}\and Sebastian Wolff\inst{3}}
\authorrunning{F. Haziza, L. Kol\'ik, R. Meyer, and S. Wolff}
\institute{\hspace{-0.5em}$^1$Uppsala University\quad $^2$Brno University of Technology\\
$^3$University of Kaiserslautern}

\begin{document}

\booktitle{Pointer Race Freedom} 
{} 
{FIT BUT Technical Report Series}
{Fr\'ed\'eric Haziza, Luk\'a\v s Hol\'ik, Roland Meyer, and Sebastian Wolff} 
{Technical Report No. FIT-TR-2015-05\\[2mm]
Faculty of Information Technology, Brno University of Technology}
{Last modified: \today}

\eject

\pagestyle{empty}

\addtocounter{page}{-2}
\pagestyle{plain}

\maketitle
\begin{abstract}
We propose a novel notion of pointer race for concurrent programs manipulating a shared heap. 
A pointer race is an access to a memory address which was freed, and it is out of the accessor's control whether or not the cell has been re-allocated. 
We establish two results.
(1) Under the assumption of pointer race freedom, 
it is sound to verify a program running under explicit memory management as if it was running with garbage collection. 
(2)
Even the requirement of pointer race freedom itself can be verified under the garbage-collected semantics.
We then prove analogues of the theorems for a stronger notion of pointer race  
needed to cope with performance-critical code purposely using racy comparisons and even racy dereferences of pointers.
As a practical contribution, we apply our results to optimize a thread-modular analysis under explicit memory management. 
Our experiments confirm a speed-up of up to two orders of magnitude. 
\end{abstract}


\section{Introduction}\label{Section:Introduction}
Today, one of the main challenges in verification is the analysis of concurrent programs that manipulate a shared heap.
The numerous interleavings among the threads make it hard to predict the dynamic evolution of the heap. 
This is even more true if explicit memory management has to be taken into account.
With garbage collection as in Java, an allocation request results in a fresh address that was not being pointed to.
The address is hence known to be owned by the allocating thread.
With explicit memory management as in C, this ownership guarantee does not hold.
An address may be re-allocated as soon as it has been freed, even if there are still pointers to it.
This missing ownership significantly complicates reasoning against the memory-managed semantics. 

In the present paper, we carefully investigate the relationship between the memory-managed semantics and the garbage-collected semantics.
We show that the difference only becomes apparent if there are programming errors of a particular form that we refer to as pointer races. 
A pointer race is a situation where a thread uses a pointer that has been freed before.
We establish two theorems.
First, if the memory-managed semantics is free from pointer races, then it coincides with the garbage-collected semantics.
Second, whether or not the memory-managed semantics contains a pointer race can be checked with the garbage-collected semantics.

The developed semantic understanding helps to optimize program analyses. 
We show that the more complicated verification of the memory-managed semantics can often be reduced to an analysis of the simpler garbage-collected semantics --- by applying the following policy: 
check under garbage collection whether the program is pointer race free. 
If there are pointer races, tell the programmer about these potential bugs. 
If there are no pointer races, rely on the garbage-collected semantics in all further analyses. 
In thread-modular reasoning, 
one of the motivations for our work, 
restricting to the garbage-collected semantics allows us to use a smaller abstract domain 
and an optimized fixed point computation.
Particularly, it removes the need to correlate the local states of threads,
and it restricts the possibilities of how threads can \mbox{influence one another}.

\begin{example}\label{Example:Treiber}
We illustrate the idea of pointer race freedom on Treiber's stack~\cite{treiber86:stack}, a lock-free implementation of a concurrent stack that provides the following methods: 
\begin{center}
\leqnomode
\scalebox{0.9}{
\hspace{4mm}
\begin{minipage}{5.75cm}
$//$\text{~global variables:~~}$\Topp$\vspace{1mm}\\
$\mathit{void}: \mathit{push}(v)$
\vspace{-0.3cm}
\begin{align}
\quad\,
&\node:=\malloc;\\
& \dsel{\node}:=v;\\
& \mathtt{repeat}\\
&\quad \topp:= \Topp;\\
&\quad \psel{\node}:=\topp;\\
&\mathtt{until}\ \mathtt{cas}(\Topp, \topp, \node);
\end{align}
\end{minipage}
\hspace{4mm}
\begin{minipage}{6.7cm}
$\mathit{bool}: \mathit{pop}(\&v)$
\vspace{-0.3cm}
\begin{align}
\qquad
&\mathtt{repeat}\\ 
&\quad \topp := \Topp;\\
&\quad \mathtt{if}\ (\topp = \mathtt{null})\ \mathtt{return}\ \mathit{false};\\
&\quad \node := \psel{\topp};\label{popsprf}\\
&\mathtt{until}\ \cas(\Topp, \topp, \node);\\
&v := \topp.data;\\
&\freeof{\topp};\ \mathtt{return}\ \mathit{true};
\end{align}
\end{minipage}}
\end{center}
This code is correct (i.e. linearizable and pops return the latest value pushed) in the presence of garbage collection, but it is incorrect under explicit memory management.
The memory-managed semantics suffers from a problem known as ABA, which indeed is related to a pointer race.
The problem arises as follows. 
Some thread $\athread$ executing pop
sets its local variable $\topp$ to the global top of the stack $\Topp$, say address $\anadr$.
The variable $\node$ is assigned the second topmost address $\anadrp$. 
While $\athread$ executes pop, another thread frees address $\anadr$ with a pop. 
Since it has been freed, address $\anadr$ can be re-allocated and pushed, becoming the top of the stack again.
However, the stack might have grown in between the free and the re-allocation.
As a consequence, $\anadrp$ is no longer the second node from the top. 
Thread $\athread$ now executes the $\cas$ (atomic compare-and-swap). 
The command first tests $\Topp = \topp$ (to check for consistency of the program state: has the top of the stack moved?).
The test passes since $\Topp$ has come back to $\anadr$ due to the re-allocation.
Thread $\athread$ then redirects $\Topp$ to $\node$. 
This is a pointer race: $\athread$ relies on the variable $\topp$ where the address was freed, and the re-allocation was not under $\athread$'s control. 
At the same time, this causes an error. 
If $\node$ no longer points to the second address from the top, moving $\Topp$ loses stack content.
\qed
\end{example}

Performance-critical implementations often intentionally make use of pointer races and employ other mechanisms to protect themselves from harmful effects due to accidental re-allocations. 
The corrected version of Treiber's stack~\cite{MichaelScottTreiber} for example equips every pointer with a version counter logging the updates. 
Pointer assignments then assign the address together with the value of the associated version counter,
and the counters are taken into account in the comparisons within $\cas$. 
That is, the $\cas(\Topp, \topp, \node)$ command atomically executes the following code:
\begin{center}
\scalebox{0.9}{
\begin{minipage}{\textwidth}
	\begin{align*}
		&\mathtt{if}\ (\Topp = \topp\ \wedge\ \Topp.\mathtt{version} = \topp.\mathtt{version})\ \{\\
		&\quad\Topp:=\node;\ \Topp.\mathtt{version}:=\topp.\mathtt{version}+1;\  \mathtt{return}\ \mathit{true}; \\
		&\}\ else\ \{\ \mathtt{return}\ \mathit{false};\ \} 
	\end{align*}
\end{minipage}}
\end{center}
\noindent
This makes the $\cas$ from Example~\ref{Example:Treiber} fail and prevents stack corruption.
Another pointer race occurs when the $\mathtt{pop}$ in Line~\eqref{popsprf} dereferences a freed pointer. 
With version counters, this is harmless.  
Our basic theory, however, would consider the comparison as well as the dereference pointer races, deeming the corrected version of Treiber's stack buggy.

To cope with performance-critical applications 
that implement version counters or techniques such as
hazard pointers \cite{Michael:HazardPointers}, reference counting \cite{Detlefs:ReferenceCounting}, or grace periods \cite{Grace2013}, we strengthen the notion of pointer race.
We let it tolerate assertions on freed pointers and dereferences of freed pointers where the value obtained by the dereference does not visibly influence the computation 
(e.g., it is assigned to a dead variable). 
To analyse programs that are only free from strong pointer races, the garbage-collected semantics is no longer sufficient.
We define a more general ownership-respecting semantics by imposing an ownership discipline on top of the memory-managed semantics.
With this semantics, we are able to show the following analogues of the above results.
First, if the program is free from strong pointer races (SPRF) under the memory-managed semantics, then the memory-managed semantics coincides with the ownership-respecting semantics.
Second, the memory-managed semantics is SPRF if and only if the ownership-respecting semantics is SPRF. 

As a last contribution, we show how to apply our theory to optimize thread-modular reasoning.   
The idea of thread-modular analysis is to buy efficiency by abstracting from the relationship between the local states of individual threads.  
The loss of precision, however, is often too severe.
For instance, any inductive invariant strong enough to show memory safety of Treiber's stack must correlate the local states of threads. 
Thread-modular analyses must compensate this loss of precision. 
Under garbage collection, an efficient way used e.g. in~\cite{GBCS:pldi07,Vafeiadis:RGSep} is keeping as a part of the local state of each thread information about the ownership of memory addresses. 
A thread owns an allocated address.  
No other thread can access it until it enters the shared part of the heap.   
Unfortunately, this exclusivity cannot be guaranteed under the memory-managed semantics. 
Addresses can be re-allocated with pointers of other threads still pointing to them. 
Works such as \cite{Sagiv:correlation,ThreadModular2013} therefore correlate the local states of threads by more expensive means (cf. Section~\ref{Section:Evaluation}), for which they pay by severely decreased scalability.  

We apply our theory to put back ownership information into thread-modular reasoning under explicit memory management. 
We measure the impact of our technique on the method of \cite{ThreadModular2013} when used to prove linearizability of programs such as Treiber's stack or Michael \& Scott's lock-free queue under explicit memory management. 
We report on resource savings of about two orders of magnitude.  

\paragraph*{Contributions}
We claim the following contributions, where $\mmsem{\aprog}$ denotes the memory-managed semantics, $\resmmsem{\aprog}$ the ownership-respecting semantics, and $\gcsem{\aprog}$ the garbage-collected semantics of program $\aprog$.\vspace{0.1cm}
\begin{itemize}
\item[(1)] We define a notion of pointer race freedom (PRF) and an equivalence $\heapequiv$ among computations such that the following two results hold.
\begin{itemize}
\item[(1.1)] If $\mmsem{\aprog}$ is PRF, then $\mmsem{\aprog}\heapequiv\gcsem{\aprog}$.\vspace{0.1cm}
\item[(1.2)] $\mmsem{\aprog}$ is PRF if and only if $\gcsem{\aprog}$ is PRF.\vspace{0.1cm} 
\end{itemize}
\item[(2)] We define a notion of strong pointer race freedom (SPRF) and an ownership-respecting semantics $\resmmsem{\aprog}$ such that the following two results hold.\vspace{0.1cm}
\begin{itemize}
\item[(2.1)] If $\mmsem{\aprog}$ is SPRF, then $\mmsem{\aprog}=\resmmsem{\aprog}$.\vspace{0.1cm}
\item[(2.2)] $\mmsem{\aprog}$ is SPRF if and only if $\resmmsem\aprog$ is SPRF.\vspace{0.1cm}
\end{itemize}
\item[(3)] 
Using the Results (2.1) and (2.2), we optimize the recent thread-modular analysis~\cite{ThreadModular2013} by a use of ownership and report on an experimental evaluation.
\end{itemize}
The Results (2.1) and (2.2) give less guarantees than (1.1) and (1.2) and hence allow for less simplifications of program analyses. 
On the other hand, the stronger notion of pointer race makes (2.1) and (2.2) applicable to a wider class of programs which would be racy in the original sense (which is the case for our most challenging benchmarks).

Finally, we note that our results are not only relevant for concurrent programs but apply to sequential programs as well. 
The point in the definition of pointer race freedom is to guarantee the following: the execution does not depend on whether a malloc has re-allocated an address, possibly with other pointers still pointing to it, or it has allocated a fresh address. However, it is mainly reasoning about concurrent programs where we see a motivation to strive for such guarantees.

\paragraph{Related Work}
Our work was inspired by the data race freedom (DRF) guarantee~\cite{AdveHillDRF1993}. 
The DRF guarantee can be understood as a contract between hardware architects and programming language designers. 
If the program is DRF under sequential consistency (SC), then the semantics on the actual architecture will coincide with SC.
We split the analogue of the statement into two, coincidence ($\mmsem{\aprog}$ PRF implies $\mmsem{\aprog}\heapequiv\gcsem{\aprog}$) and means of checking ($\mmsem{\aprog}$ PRF iff $\gcsem{\aprog}$ PRF).
There are works that weaken the DRF requirement while still admitting efficient analyses~\cite{Owens2010,AlglaveM11,BMM11}.
Our notion of strong pointer races is along this line.

The closest related work is~\cite{Grace2013}.
Gotsman et al. study re-allocation under explicit memory management.
The authors focus on lock-free data structures implemented with hazard pointers, read-copy-update, or epoch-based reclamation. 
The key observation is that all three techniques rely on a common synchronization pattern called grace periods.
Within a grace period of a cell $\anadr$ and a thread $\athread$, the thread can safely access the cell without having to fear a free command.
The authors give thread-modular reasoning principles for grace periods and show that they lead to elegant and scalable proofs. 

The relationship with our work is as follows.
If grace periods are respected, then the program is guaranteed to be SPRF (there are equality checks on freed addresses). 
Hence, using Theorem~\ref{Theorem:SPRFGuarantee} in this work, 
it is sufficient to verify lock-free algorithms under the ownership-respecting semantics.
Interestingly, Gotsman et al. had an intuitive idea of pointer races without making the notion precise (quote: \emph{...potentially harmful race between threads accessing nodes and those trying to reclaim them is avoided}~\cite{Grace2013}).
Moreover, they did not study the implications of race freedom on the semantics, which is the main interest of this paper.
We stress that our approach does not make assumptions about the synchronization strategy. 
Finally, Gotsman et al. do not consider the problem of checking the synchronization scheme required by grace periods.
We show that PRF and SPRF can actually be checked on simpler semantics.

Data refinement in the presence of low-level memory operation is studied in~\cite{DataRefinement2005}.
The work defines a notion of substitutability that only requires a refinement of \emph{error-free computations}. 
In particular, there is no need to refine computations that dereference dangling pointers.
In our terms, these dereferences yield pointer races.
We consider~\cite{DataRefinement2005} as supporting our requirement for (S)PRF.

The practical motivation of our work, 
thread-modular analysis~\cite{FlQa:spin03}, has already been discussed. 
We note the adaptation to heap-manipulating programs~\cite{GBCS:pldi07}. 
Interesting is also the combination with separation logic from \cite{Vafeiadis:RGSep,Vafeiadis:vmcai09} (which uses ownership to improve precision). 
There are other works studying shape analysis and thread-modular analysis. 
As these fields are only a part of the motivation, we do not provide a full overview.          
\section{Heap-manipulating Programs}\label{Section:Programs}
\subsubsection*{Syntax}
We consider \bfemph{concurrent heap-manipulating programs}, defined to be sets of threads $\aprog = \set{\athread_1, \athread_2, \ldots}$ from a set $\threads$.
We do not assume finiteness of programs.
This ensures our results carry over to programs with a parametric number of threads.
Threads $\athread$ are ordinary while-programs operating on data and pointer variables.
Data variables are denoted by $\advar, \advarp\in\dvars$.
For pointer variables, we use $\apvar, \apvarp\in\pvars$.
We assume $\dvars\cap \pvars = \emptyset$ and obey this typing.
Pointer variables come with selectors $\pselarg{\apvar}{1},\ldots, \pselarg{\apvar}{n}$ and $\dsel{\apvar}$ for finitely many pointer fields and one data field (for simplicity; the generalization to arbitrary data fields is straightforward).
We use $\apt$ to refer to pointers $\apvar$ and $\psel{\apvar}$.
Similarly, by $\adt$ we mean data variables $\advar$ and the corresponding selectors $\dsel{\apvar}$.
Pointer and data variables are either \emph{local} to a thread, indicated by $\apvar,\advar\in\localof{\athread}$, or they are \emph{shared} among the threads in the program.
We use $\shared$ for the set of all shared variables.

The \bfemph{commands} $\acom\in \coms$ employed in our while-language are
\begin{align*}
\acond&::=\phantom{\bnf} \apvar = \apvarp\bnf \advar = \advarp\bnf \neg \acond\\
\acom&::=\phantom{\bnf} \assert\ \acond \bnf \apvar:= \malloc \bnf \freeof{\apvar}\\
&\phantom{::=}\bnf \apvarp := \psel{\apvar}\bnf \psel{\apvar}:=\apvarp\bnf \apvar:=\apvarp\\
&\phantom{::=}\bnf \advar:=\dsel{\apvar}\bnf
\dsel{\apvar}:=\advar\bnf \advar:=\opof{\advar_1,\ldots, \advar_n}\ .
\end{align*}
Pointer variables are allocated with $\apvar := \malloc$ and freed via  $\freeof{\apvar}$.
Pointers and data variables can be used in assignments.
These assignments are subject to typing: we only assign pointers to pointers and data to data.
Moreover, a thread only uses shared variables and its own local variables.
To compute on data variables, we support operations $\op$ that are not specified further.
We only assume that the program comes with a data domain $(\dom, \ops)$ so that its operations $\op$ stem from $\ops$.
We support assertions that depend on equalities and inequalities among pointers and data variables.
Like in if and while commands, we require assertions to have complements: if a control location has a command $\assert\ \acond$, then it also has a command $\assert\ \neg \acond$. 
We use as a running example the program in Example~\ref{Example:Treiber}, Treiber's stack \cite{treiber86:stack}.
\subsubsection*{Semantics}
A heap is defined over a set of addresses $\adr$ that contains the distinguished element $\segval$. 
Value $\segval$ indicates that a pointer has not yet been assigned a cell and thus its data and next selectors cannot be accessed.
Such an access would result in a segfault.
A heap gives the valuation of pointer variables $\pvars\nrightarrow \adr$,
the valuation of the next selector functions $\adr\nrightarrow \adr$, the valuation of the data variables $\dvars\nrightarrow \dom$, and the valuation of the data selector fields $\adr\nrightarrow \dom$.
In Section~\ref{Section:PRF}, we will restrict heaps to a subset of so-called valid pointers.
To handle such restrictions, it is convenient to let heaps evaluate expressions $\psel{\anadr}$ rather than next functions.
Moreover, with the use of restrictions valuation functions will typically be partial. 

Let $\pexp:=\pvars\disunion \setcond{\psel{\anadr}}{\anadr\in\adr\setminus\set{\segval}\text{ and } \mathtt{next}\text{ a selector}}$ be the set of pointer expressions and $\dexp:=\dvars\disunion \setcond{\dsel{\anadr}}{\anadr\in\adr\setminus\set{\segval}}$ be the set of data expressions. 
A \bfemph{heap} is a pair 
$\aheap = (\apval, \adval)$ 
with $\apval:\pexp\nrightarrow \adr$ the valuation of the pointer expressions and $\adval:\dexp\nrightarrow \dom$ the valuation of the data expressions.
We use $\apexp$ and $\adexp$ for a pointer and a data expression, and also write $\aheap(\apexp)$ or $\aheap(\adexp)$.
The valuation functions are clear from the expression.
The addresses inside the heap that are actually in use are
\begin{align*}
\adrof{\aheap}:=(\domof{\apval}\cup \rangeof{\apval}\cup \domof{\adval})\cap \adr.
\end{align*} 
Here, we use $\set{\psel{\anadr}}\cap \adr :=\set{\anadr}$ and similar for data selectors.

We model heap modifications with \bfemph{updates} $[\apexp\mapsto \anadr]$ and $[\adexp\mapsto \advalue]$ from the set $\updates$.
Update $[\apexp\mapsto \anadr]$ turns the partial function $\apval$ into the new partial function $\apval[\apexp\mapsto \anadr]$ with $\domof{\apval[\apexp\mapsto\anadr]}:=\domof{\apval}\cup\set{\apexp}$.
It is defined by $\apval[\apexp\mapsto\anadr](\apexpp):=\apval(\apexpp)$ if $\apexpp\neq \apexp$, and $\apval[\apexp\mapsto\anadr](\apexp):=\anadr$.
We also write $\aheap[\apexp\mapsto \anadr]$ since the valuation that is altered is clear from the update.

We define three semantics for concurrent heap-manipulating programs.
All three are in terms of computations, sequences of actions from
$\actions\ :=\ \threads\times \coms \times \updates$. 
An action $\anact = (\athread, \acom, \anup)$ consist of a thread $\athread$, a command $\acom$ executed in the thread, and an update $\anup$.
By $\threadof{\anact}:=\athread$, $\comof{\anact}:=\acom$, and $\updateof{\anact}:=\anup$ we access the thread, the command, and the update in $\anact$. 
To make the heap resulting from a computation $\tau\in\actions^*$ explicit, we define $\heapcomput{\varepsilon}:=(\emptyset, \emptyset)$ and 
$\heapcomput{\tau.\anact} := \heapcomput{\tau}[\updateof{\anact}]$. 
So we modify the current heap with the update required by the last action.

The garbage-collected semantics and the memory-managed semantics only differ on allocations.
We define a strict form of garbage collection that never re-allocates a cell.
With this, we do not have to define unreachable parts of the heap that should be garbage collected. 
We only model computations that are free from segfaults.
This means a transition accessing next and data selectors is enabled only if the corresponding pointer is assigned a cell.

Formally, the \bfemph{garbage-collected semantics} of a program~$\aprog$, denoted by~$\gcsem{\aprog}$, is a set of computations in $\actions^*$. 
The definition is inductive. 
In the base case, we have single actions $(\bot, \bot, [\apval, \adval])\in\gcsem{\aprog}$ with 
$\apval:\pvars\rightarrow\set{\segval}$ and
$\adval:\dvars\rightarrow \dom$ arbitrary.
No pointer variable is mapped to a cell and the data variables contain arbitrary values. 
In the induction step, consider $\tau\in\gcsem{\aprog}$ where thread $\athread$ is ready to execute command $\acom$.
Then $\tau.(\athread, \acom, \anup)\in\gcsem{\aprog}$, provided one of the following rules holds. 
\begin{description}
\item[(Asgn)] Let $\acom$ be $\psel{\apvar}:=\apvarp$, $\heapcomputof{\tau}{\apvar}=\anadr\neq \segval$, $\heapcomputof{\tau}{\apvarp}=\anadrp$. 
We set $\anup=[\psel{\anadr}\mapsto\anadrp]$. 
The remaining assignments are similar.\vspace{0.1cm}
\item[(Asrt)] Let $\acom$ be $\assert\ \apvar=\apvarp$. 
The precondition is $\heapcomputof{\tau}{\apvar}=\heapcomputof{\tau}{\apvarp}$.
There are no updates, $\anup=\emptyset$.
The assertion with a negated condition is defined analogously.
A special case occurs if $\heapcomputof{\tau}{\apvar}$ or $\heapcomputof{\tau}{\apvarp}$ is $\segval$.
Then the assert and its negation will pass.
Intuitively, undefined pointers hold arbitrary values. 
Our development does not depend on this modeling choice.\vspace{0.1cm}
\item[(Free)] If $\acom$ is $\freeof{\apvar}$, there are no constraints and no updates. \vspace{0.1cm}
\item[(Malloc1)] Let $\acom$ be $\apvar:=\malloc$, $\anadr\notin\adrof{\heapcomput{\tau}}$, and $\advalue\in \dom$.
Then we define $\anup=[\apvar\mapsto \anadr, \dsel{\anadr}\mapsto\advalue,\setcond{\psel{\anadr}\mapsto\segval}{\text{for every selector $\mathtt{next}$}}]$. The rule only allocates cells that have not been used in the computation. 
Such a cell holds an arbitrary data value and the next selectors have not yet been allocated.
\end{description}

With explicit memory management, we can re-allocate a cell as soon as it has been freed.
Formally, the \bfemph{memory-managed semantics} $\mmsem{\aprog}\subseteq \actions^*$ is defined like $\gcsem{\aprog}$ but has a second allocation rule:
\begin{description}
\item[(Malloc2)] Let $\acom$ be $\apvar:=\malloc$ and $\anadr\in \freedof{\tau}$. 
Then $\anup=[\apvar\mapsto \anadr]$. 
\end{description}
Note that (Malloc2) does not alter the selectors of address $\anadr$.
The set $\freedof{\tau}$ contains the addresses that have been allocated in $\tau$ and freed afterwards.
The definition is by induction.
In the base case, we have $\freedof{\varepsilon}:= \emptyset$.
The step case is
\begin{align*}
\freedof{\tau.(\athread, \freeof{\apvar}, \anup)}&:=\freedof{\tau}\cup \set{\anadr},&&\text{if $\heapcomputof{\tau}{\apvar}=\anadr\neq \segval$}\\
\freedof{\tau.(\athread, \apvar:=\malloc, \anup)}&:=\freedof{\tau}\setminus \set{\anadr},&&\text{if $\malloc$ returns $\anadr$}\\
\freedof{\tau.(\athread, \anact, \anup)}&:=\freedof{\tau},&&\text{otherwise}.
\end{align*}
\section{Pointer Race Freedom}\label{Section:PRF}
We show that for well-behaved programs the garbage-collected semantics coincides with the memory-managed semantics.
Well-behaved means there is no computation where one pointer frees a cell and later a dangling pointer accesses this cell.
We call such a situation a \bfemph{pointer race}, referring to the fact that the free command and the access are not synchronized, for otherwise the access should have been avoided. 
To apply this equivalence, we continue to show how to reduce the check for pointer race freedom itself to the garbage-collected semantics.
\subsection{PRF Guarantee}
The definition of pointer races relies on a notion of validity for pointer expressions.
To capture the situation sketched above, a pointer is invalidated if the cell it points to is freed.
A pointer race is now an access to an invalid pointer.
The definition of validity requires care when we pass pointers.
Consider an assignment $\apvar:=\psel{\apvarp}$ where $\apvarp$ points to $\anadr$ and $\psel{\anadr}$ points to $\anadrp$.
Then $\apvar$ becomes a valid pointer to $\anadrp$ only if both $\apvarp$ and $\psel{\anadr}$ were valid.
In Definition~\ref{Definition:Validity}, we use $\apexp$ to uniformly refer to $\apvar$ and $\psel{\anadr}$ on the left-hand side of assignments.
In particular, we evaluate pointer variables $\apvar$ to $\heapcomputof{\tau}{\apvar}=\anadr$ and write $\psel{\anadr}:=\apvarp$ for the assignment $\psel{\apvar}:=\apvarp$.
\begin{definition}\label{Definition:Validity}
The \bfemph{valid} pointer expressions in a computation $\tau\in\mmsem{\aprog}$, denoted by $\validof{\tau}\subseteq \pexp$, are defined inductively by $\validof{\varepsilon}:= \pexp$ and
\begin{align*}
\validof{\tau.(\athread, \apvar:=\psel{\apvarp}, \anup)} &:= \validof{\tau}\cup\set{\apvar},&&\text{if }\apvarp\in \validof{\tau}\wedge \psel{\heapcomputof{\tau}{\apvarp}}\in\validof{\tau}\\
\validof{\tau.(\athread, \apvar:=\psel{\apvarp}, \anup)} &:= \validof{\tau}\setminus\set{\apvar},&&\text{if }
\apvarp\notin \validof{\tau}\vee \psel{\heapcomputof{\tau}{\apvarp}}\notin\validof{\tau}\\
\validof{\tau.(\athread, \apexp:=\apvarp, \anup)} &:= \validof{\tau}\cup\set{\apexp},&&\text{if }\apvarp\in \validof{\tau}\\
\validof{\tau.(\athread, \apexp:=\apvarp, \anup)} &:= \validof{\tau}\setminus\set{\apexp},&&\text{if }\apvarp\notin \validof{\tau}\\
\validof{\tau.(\athread, \freeof{\apvar}, \anup)}&:=\validof{\tau}\setminus \invalidof{\anadr},&&\text{if }\anadr=\heapcomputof{\tau}{\apvar}\\
\validof{\tau.(\athread, \apvar:=\malloc, \anup)}&:=\validof{\tau}\cup\set{\apvar},\\
\validof{\tau.(\athread, \anact, \anup)}&:=\validof{\tau},&&\text{otherwise.}
\end{align*}
If $\anadr\neq \segval$, then $\invalidof{\anadr}:=\setcond{\apexp}{\heapcomputof{\tau}{\apexp}=\anadr}\cup \set{\pselarg{\anadr}{1},\ldots, \pselarg{\anadr}{n}}$. 
If $\anadr= \segval$, then $\invalidof{\anadr}:=\emptyset$.
\end{definition}
When we pass a valid pointer, this validates the receiver (adds it to $\validof\tau$).
When we pass an invalid pointer, this invalidates the receiver.
As a result, only some selectors of an address may be valid.
When we free an address $\anadr\neq \segval$, all expressions that point to $\anadr$ as well as all next selectors of $\anadr$ become invalid.
This has the effect of isolating $\anadr$ so that the address behaves like a fresh one for valid pointers.
A malloc validates the respective pointer but does not validate the next selectors of the allocated address. 

\begin{definition}[Pointer Race]\label{Definition:PRF}
A computation $\tau.(\athread, \acom, \anup)\in\mmsem{\aprog}$ is called a \bfemph{pointer race (PR)}, if $\acom$ is
\begin{itemize}
\item[(i)] a command containing $\dsel{\apvar}$ or $\psel{\apvar}$ or $\freeof{\apvar}$ with $\apvar\notin \validof{\tau}$, or 
\item[(ii)] an assertion containing $\apvar\notin \validof{\tau}$.
 \end{itemize}
\end{definition}
The last action of a PR is said to \bfemph{raise a PR}. 
A set of computations is \bfemph{pointer race free (PRF)} if it does not contain a PR. 
In Example~\ref{Example:Treiber}, 
the discussed comparison $\Topp=\topp$ within $\cas$ raises a PR since $\topp$ is invalid. 
It is worth noting that we can still pass around freed addresses without raising a PR.
This means the memory-managed and the garbage-collected semantics will not yield isomorphic heaps, but only yield isomorphic heaps on the valid pointers. 
We now define the notion of isomorphism among heaps $\aheap$.  

A function $\funa:\adrof{\aheap}\rightarrow \adr$ is an address mapping, if $\funa(\anadr)=\segval$ if and only if $\anadr=\segval$. 
Every address mapping induces a function $\fune:\domof{\aheap}\rightarrow\pexp\cup \dexp$ on the pointer and data expressions inside the heap by
\begin{alignat*}{3}
\fune(\apvar)&:=\apvar&\hspace{2cm} \fune(\advar)&:=\advar\\
\fune(\psel{\anadr})&:=\psel{\funa(\anadr)} &
\fune(\dsel{\anadr})&:=\dsel{\funa(\anadr)}.
\end{alignat*}
Pointer and data variables are mapped identically.
Pointers on the heap $\psel{\anadr}$ are mapped to $\psel{\funa(\anadr)}$ as defined by the address mapping, and similar for the data.
\begin{definition}
Two heaps $\aheap_1$ and $\aheap_2$ with $\aheap_i=(\apval_i, \adval_i)$ are \bfemph{isomorphic}, 
denoted by $\aheap_1\heapiso \aheap_2$, if there is a bijective address mapping
$\anisoa: \adrof{\aheap_1}\rightarrow \adrof{\aheap_2}$ where the induced 
$\anisoe: \domof{\aheap_1}\rightarrow \domof{\aheap_2}$ is again bijective and satisfies
\begin{alignat*}{3}
\anisoa(\apval_1(\apexp)) &= \apval_2(\anisoe(\apexp))&\qquad 
\adval_1(\adexp) &= \adval_2(\anisoe(\adexp)).
\end{alignat*}
\end{definition}

To prove a correspondence between the two semantics, we restrict heaps to the valid pointers.
The restriction operation keeps the data selectors for all addresses that remain. 
To be more precise, let $\aheap=(\apval, \adval)$ and $P\subseteq \pexp$.
The \bfemph{restriction of $\aheap$ to $P$} is the new heap 
$\restrict{\aheap}{P}:=(\restrict{\apval}{P}, \restrict{\adval}{D})$ with
\begin{align*}
D:=\dvars\cup \setcond{\dsel{\anadr}}{\exists \apexp\in \domof{\apval}\cap P:\apval(\apexp)=\anadr}\ .
\end{align*}
Restriction and update enjoy a pleasant interplay with isomorphism. 
\begin{lemma}\label{Lemma:HeapIsoAlgebra}
Let $\aheap_1\heapiso \aheap_2$ via $\anisoa$ and let $P\subseteq \pexp$.
Then
\begin{align}
\restrict{\aheap_1}{P}&\heapiso \restrict{\aheap_2}{\anisoe(P)}\label{Equation:HeapIsoRestrict}\\
\aheap_1[\psel{\anadr}\mapsto\anadrp]&\heapiso \aheap_2[\psel{\anadr'}\mapsto \anadrp']\label{Equation:HeapIsoModifyPointer}\\
\aheap_1[\dsel{\anadr}\mapsto d]&\heapiso \aheap_2[\dsel{\anadr'}\mapsto d]\label{Equation:HeapIsoModifyData}.
\end{align}
Isomorphisms~\eqref{Equation:HeapIsoModifyPointer} and~\eqref{Equation:HeapIsoModifyData} have a side condition.
If $\anadr\in\adrof{\aheap_1}$ then $\anadr'=\anisoa(\anadr)$.
If $\anadr\notin\adrof{\aheap_1}$ then $\anadr'\notin\adrof{\aheap_2}$,
and similar for $\anadrp$.
\end{lemma}

Two computations are heap equivalent, if their sequences of actions coincide when projected to the threads and commands, and if the resulting heaps are isomorphic on the valid part.
We use $\downarrow$ for projection.
\begin{definition}
Computations $\tau, \sigma\in\mmsem{\aprog}$ are \bfemph{heap-equivalent}, $\tau\heapequiv \sigma$, if
\begin{align*}
\project{\tau}{\threads\times\coms}\ =\ \project{\sigma}{\threads\times\coms}\qquad
\text{and}\qquad \restrict{\heapcomput{\tau}}{\validof{\tau}}&\heapiso\ \restrict{\heapcomput{\sigma}}{\validof{\sigma}}\ .
\end{align*}
\end{definition}
We also write $\mmsem{\aprog}\heapequiv\gcsem{\aprog}$ to state that for every computation $\tau\in\mmsem{\aprog}$, there is a computation $\sigma\in\gcsem{\aprog}$ with $\tau\heapequiv \sigma$, and vice versa.

We are now ready to state the PRF guarantee. 
The idea is to consider pointer races programming errors.
If a program has pointer races, the programmer should be warned.
If the program is PRF, further analyses can rely on the garbage-collected semantics:
\begin{theorem}[PRF Guarantee]\label{Theorem:PRFGuarantee}
If $\mmsem{\aprog}$ is PRF, then $\mmsem{\aprog}\heapequiv \gcsem{\aprog}$.
\end{theorem}
The memory-managed semantics of Treiber's stack suffers from the ABA-problem while the garbage-collected semantics does not. 
The two are not heap-equivalent. 
By Theorem~\ref{Theorem:PRFGuarantee}, the difference is due to a PR. 
One such race is discussed in Example~\ref{Example:Treiber}.


In the proof of Theorem~\ref{Theorem:PRFGuarantee},
the inclusion from right to left always holds. 
The reverse direction needs information about the freed addresses:
if an address has been freed, it no longer occurs in the valid part of the heap --- 
provided the computation is PRF.
\begin{lemma}\label{Lemma:Freed}
Assume $\tau\in \mmsem{\aprog}$ is PRF.
Then $\freedof{\tau}\cap \adrof{\restrict{\heapcomput{\tau}}{\validof{\tau}}}=\emptyset$.
\end{lemma}

Lemma~\ref{Lemma:HeapIsoAlgebra} and~\ref{Lemma:Freed} allow us to prove 
Proposition~\ref{Proposition:PRFImpliesGC}.
The result implies the missing direction of Theorem~\ref{Theorem:PRFGuarantee} and will also be helpful later on.
\begin{proposition}\label{Proposition:PRFImpliesGC}
Consider $\tau\in\mmsem{\aprog}$ PRF. Then there is $\sigma\in\gcsem{\aprog}$ with $\sigma\comequiv \tau$.
\end{proposition}

To apply Theorem~\ref{Theorem:PRFGuarantee}, one has to prove $\mmsem{\aprog}$ PRF. 
We develop a technique for this.
\subsection{Checking PRF}\label{Section:CheckPRF}
We show that checking pointer race freedom for the memory-managed semantics can be reduced to checking pointer race freedom for the garbage-collected semantics. 
The key argument is that the earliest possible PR always lie in the garbage-collected semantics.
Technically, we consider a shortest PR in the memory-managed semantics and from this construct a PR in the garbage-collected semantics. 
\begin{theorem}[Checking PRF]\label{Theorem:CheckPRF}
$\mmsem{\aprog}$ is PRF if and only if $\gcsem{\aprog}$ is PRF.
\end{theorem}
To illustrate the result, the pointer race in Example~\ref{Example:Treiber} belongs to the memory-managed semantics. 
Under garbage collection, there is a similar computation which does not re-allocate $\anadr$. 
Freeing $\anadr$ still renders $\topp$ invalid 
and, as before, leads to a PR in $\cas$. 
The proof of Theorem~\ref{Theorem:CheckPRF} applies Proposition~\ref{Proposition:PRFImpliesGC} to mimic the shortest racy computation up to the last action.
To mimic the action that raises the PR, we need the fact that an invalid pointer variable does not hold $\segval$, as stated in the following lemma.
\begin{lemma}\label{Lemma:GCPRF}
Consider a PRF computation $\sigma\in\gcsem{\aprog}$.
(i) If $\apvar\notin\validof{\sigma}$, then 
$\heapcomputof{\sigma}{\apvar}\neq \segval$.
(ii) If $\apexp\in\validof{\sigma}$, $\heapcomputof{\sigma}{\apexp}=\anadr\neq \segval$, and $\psel{\anadr}\notin\validof{\sigma}$, then $\heapcomputof{\sigma}{\psel{\anadr}}\neq \segval$.
\end{lemma}

While the completeness proof of Theorem~\ref{Theorem:CheckPRF} is non-trivial, checking PRF for $\gcsem{\aprog}$ is an easy task.  
One instruments the given program $\aprog$ to a new program $\aprog'$ as follows:  $\aprog'$ tags every address that is freed and checks whether a tagged address is dereferenced, freed, or used in an assertion. 
In this case, $\aprog'$ enters a distinguished goal state.
\begin{proposition}\label{Proposition:Instrumentation}
$\gcsem{\aprog}$ is PRF if and only if $\gcsem{\aprog'}$ cannot reach the goal state.  
\end{proposition}
For the correctness proof, we only need to observe that under garbage collection the invalid pointers are precisely the pointers to the freed cells.
\begin{lemma}\label{Lemma:GCValid}
Let $\sigma\in\gcsem{\aprog}$ and $\heapcomputof{\sigma}{\apexp}=\anadr\neq \segval$.
Then $\apexp\notin\validof{\sigma}$ iff $\anadr\in\freedof{\sigma}$.
\end{lemma}
The lemma does not hold for the memory-managed semantics.
Moreover, the statement turns Lemma~\ref{Lemma:Freed}, which can be read as an implication, into an equivalence. 
Namely, Lemma~\ref{Lemma:Freed} says that if a pointer has been freed, then it cannot be valid. 
Under the assumtpions of Lemma~\ref{Lemma:GCValid}, it also holds that if a pointer is not valid, then it has been freed.
\section{Strong Pointer Race Freedom}\label{Section:SPRF}
The programing style in which a correct program should be pointer race free
counts on the following policy:
a memory address is freed only if it is not meant to be touched until its re-allocation, by any means possible.

This simplified treatment of dynamic memory is practical in common programing tasks, 
but the authors of performance-critical applications are often forced to employ subtler techniques.  
For example, the version of Treiber's stack equipped with version counters to prevent ABA under explicit memory management contains two violations of the simple policy, both of which are pointer races.   
(1) The $\cas$ may compare invalid pointers. 
This could potentially lead to ABA, but the programmer prevents the harmful effect of re-allocation using version counters, which make the $\cas$ fail.
(2) The command $\node:=\psel\topp$ in Line~\eqref{popsprf} of $\mathtt{pop}$ may dereference the $\mathtt{next}$ field of a freed (and therefore invalid) pointer. 
This is actually correct only under the assumption 
that neither the environment nor any thread of the program itself may redirect a once valid pointer outside the accessible memory (otherwise the dereference could lead to a segfault).
The value obtained by the dereference may again be influenced by that the address was re-allocated. 
The reason for why this is fine is that the subsequent $\cas$ is bound to fail, which makes $\node$ a dead variable --- its value does not matter.

In both cases, the programmer only prevents side effects of an accidental re-allocation. 
He uses a subtler policy and frees an address only if 
its \emph{content} is not meant to be of any relevance any more. 
Invalid addresses can still be compared, and their pointer fields can even be dereferenced unless the obtained value influences the control.

\subsection{SPRF Guarantee}

We introduce a stronger notion of pointer race that expresses the above subtler policy. In the definition, we will call strongly invalid the pointer expressions that
have obtained their value from dereferencing an invalid/freed pointer.
\begin{definition}[Strong Invalidity]\label{Definition:StrongInvalidity}
The set of \bfemph{strongly invalid} expressions in $\tau\in\mmsem{\aprog}$, denoted by $\sinvalidof{\tau}\subseteq \pexp\cup \dexp$, is defined inductively by $\sinvalidof{\varepsilon}:= \emptyset$ and
\begin{align*}
\sinvalidof{\tau.(\athread, \apvar:=\psel{\apvarp}, \anup)} &:= \sinvalidof{\tau}\cup\set{\apvar},&&\text{if }\apvarp\not\in \validof{\tau}\\
\sinvalidof{\tau.(\athread, \apexp:=\apvarp, \anup)} &:= \sinvalidof{\tau}\cup\set{\apexp},&&\text{if }\apvarp\in \sinvalidof{\tau}\\
\sinvalidof{\tau.(\athread, \advar:=\dsel{\apvarp}, \anup)} &:= \sinvalidof{\tau}\cup\set{\advar},&&\text{if }\apvarp\not\in \validof{\tau}\\
\sinvalidof{\tau.(\athread, \adexp:=\advar, \anup)} &:= \sinvalidof{\tau}\cup\set{\adexp},&&\text{if }\advar\in \sinvalidof{\tau}\\
\sinvalidof{\tau.\anact} &:= \sinvalidof{\tau}\setminus\validof{\tau.\anact}, &&\text{in all other cases.}
\end{align*}
\end{definition}
The value obtained by dereferencing a freed pointer may depend on 
actions of other threads that point to the cell due to re-allocation.
However, by assuming that a once valid pointer can never be set to $\segval$, 
we obtain a guarantee that the actions of other threads cannot prevent the dereference itself from being executed (they cannot make it segfault). 
Assigning the uncontrolled value to a local variable is therefore not harmful.
We only wish to prevent a correct computation from being influenced by that value. 
We thus define incorrect/racy any attempt to compare or dereference the value. 
Then, besides allowing for the creation of strongly invalid pointers, the notion of strong pointer race strengthens PR by tolerating comparisons of invalid pointers.
\begin{definition}[Strong Pointer Race]\label{Definition:SPRF} 
A computation $\tau.(\athread, \acom, \anup)\in\mmsem{\aprog}$
is a \bfemph{strong pointer race (SPR)}, if the command $\acom$ is one of the following: 
\begin{enumerate}

\item[(i)] $\psel{\apvar}:=\apvarp$ or $\dsel{\apvar}:=\advar$ or $\freeof{\apvar}$ with $\apvar\notin\validof{\tau}$

\item[(ii)] an assertion containing $\apvar$ or $\advar$ in $\sinvalidof\tau$ 
\item[(iii)] a command containing $\psel{\apvar}$ or $\dsel{\apvar}$ where $\apvar\in\sinvalidof\tau$.
\end{enumerate}
\end{definition}
The last action of an SPR \bfemph{raises an SPR}. 
A set of computations is \bfemph{strong pointer race free (SPRF)} if it does not contain an SPR. 
An SPR can be seen in Example~\ref{Example:Treiber} as a continuation of the race ending at $\cas$.
The subsequent $\topp:=\free$ raises an SPR as $\topp$ is invalid.
The implementation corrected with version counters is SPRF.

Theorems~\ref{Theorem:PRFGuarantee} and \ref{Theorem:CheckPRF} no longer hold for strong pointer race freedom. 
It is not possible to verify $\mmsem\aprog$ modulo SPRF by analysing $\gcsem \aprog$.
The reason is that the garbage-collected semantics does not cover SPRF computations that compare or dereference invalid pointers. 
To formulate a sound analogy of the theorems, we have to replace $\gcsem .$  by a more powerful semantics. 
This, however, comes with a trade-off.
The new semantics should still be amenable to efficient thread-modular reasoning.

The idea of our new semantics $\resmmsem{\aprog}$ is to introduce the concept of ownership to the memory-managed semantics, and show that SPRF computations stick to it.
Unlike with garbage collection, we cannot use a naive notion of ownership that guarantees the owner exclusive access to an address.
This is too strong a guarantee. 
In $\mmsem\aprog$, 
other threads may still have access to an owned address via invalid pointers.
Instead, we design ownership such that dangling pointers are not allowed to influence the owner. 
The computation will thus proceed as if the owner had allocated a fresh address.

To this end, we let a thread own an allocated address until one of the two events happen: either (1) the address is \emph{published}, that is, it enters the shared part of the heap (which consists of addresses reached from shared variables by following valid pointers and of freed addresses), or
(2)  the address is \emph{compromised}, that is,
the owner finds out that the cell is not fresh by comparing it with an invalid pointer. 
Taking away ownership in this situation is needed since the owner can now change its behavior based on the re-allocation. 
The owner may also spread the information about the re-allocation among the other threads and change their behavior, too. 
It can thus no longer be guaranteed that the computation will continue as if a fresh address had been allocated.

\begin{definition}[Owned Addresses]\label{Definition:Owned}
For $\tau\in\mmsem{\aprog}$ and a thread $\athread$,
we define the set of addresses \bfemph{owned by $\athread$}, denoted by $\ownedof \athread \tau$, as $\ownedof \athread{\varepsilon}:=\emptyset$ and
\begin{align*}
\ownedof \athread{\tau.(\athread,\apvar:=\malloc,\anup)}&\!:=\!
\ownedof \athread{\tau}\!\cup\!\set{\anadr}, \!\!\!\!\!&&\text{if } \apvar\in \localof{\athread}\text{ and $\malloc$ returns $\anadr$}\\
\ownedof \athread{\tau.(\athread,\freeof{\apvar},\emptyset)}&\!:=\! \ownedof\athread{\tau}\!\setminus\! \set{\heapcomputof{\tau}{\apvar}}, \!\!\!\!\!&&\text{if }\apvar\in\validof{\tau}\\
\ownedof \athread{\tau.(\athread,\apvar:=\apvarp,[\apvar\mapsto\anadr])}&\!:=\!
\ownedof \athread{\tau}\!\setminus\!\set{\anadr}, \!\!\!\!\!&&\text{if } \apvar\in \shared\wedge\apvarp\in\validof{\tau}\\
\ownedof \athread{\tau.(\athread,\apvar:=\psel{\apvarp},[\apvar\mapsto\anadr])}&\!:=\!
\ownedof \athread{\tau}\!\setminus\!\set{\anadr}, \!\!\!\!\!&&\text{if } \apvar\in \shared\wedge\apvarp,\psel{\heapcomputof{\tau}{\apvarp}}\in\validof{\tau}\\
\ownedof \athread{\tau.(\cdot,\apvar:=\psel\apvarp,[\apvar\mapsto \anadr])}&\!:=\! \ownedof \athread{\tau}\!\setminus\!\set{\anadr}, \!\!\!\!\!&&\text{if } \heapcomputof{\tau}{\apvarp}\!\!\not\in\!\ownedof\athread\tau\wedge \apvarp,\psel{\heapcomputof{\tau}{\apvarp}}\!\!\in\!\validof{\tau}\\
\ownedof \athread{\tau.(\athread,\assert\ \apvar=\apvarp,\emptyset)}&\!:=\! \ownedof \athread{\tau}\!\setminus\!\set{\heapcomputof{\tau}{\apvar}}, \!\!\!\!\!&&\text{if } \apvar\notin\validof{\tau}\vee \apvarp\notin\validof{\tau}\\
\ownedof {\athread}{\tau.\anact} &\!:=\! \ownedof {\athread}{\tau}, \!\!\!\!\!&&\text{in all other cases}.
\end{align*}
\end{definition}
The first four cases of losing ownership are due to publishing, the last case is due to the address being compromised by comparing with an invalid pointer.

The following lemma states the intuitive fact that an owned address cannot be pointed to by a valid shared variable or by a valid local variable of another thread, since such a configuration can be achieved only by publishing the address. 
\begin{lemma}\label{Lemma:OwnImpliesLocal}
Let $\tau\in\mmsem{\aprog}$ and $\apvar\in\validof{\tau}$ with $\heapcomputof{\tau}{\apvar}\in\ownedof{\athread}{\tau}$. 
Then $\apvar\in\localof{\athread}$.
\end{lemma}

We now define ownership violations as precisely those situations in which the fact that an owned address was re-allocated while an invalid pointer was still pointing to it influences the computation. 
Technically, 
the address is freed or its content is altered due to
an access via a pointer of another thread or a shared pointer.
\begin{definition}[Ownership Violation]\label{Definition:OwnershipViolation}
A computation $\tau.(\athread, \acom, \anup)\in\mmsem{\aprog}$ \bfemph{violates ownership}, if $\acom$ is one of the following
\begin{align*}
\psel{\apvarp}:=\apvar,\quad\dsel{\apvarp}:=\advar,\quad\text{or}\quad \freeof{\apvarp},
\end{align*}
where $\heapcomputof{\tau}{\apvarp}\in\ownedof{\athread'}{\tau}$ and ($\athread'\neq \athread$ or $\apvarp\in\shared$).
\end{definition}
The last action of a computation violating ownership is called an \bfemph{ownership violation}
and a computation which does not violate ownership \bfemph{respects ownership}.
We define the \bfemph{ownership-respecting semantics} $\resmmsem \aprog$ as those computations of $\mmsem \aprog$ that respect ownership.
The following lemma shows that SPRF computations respect ownership. 

\begin{lemma}\label{Lemma:OwnershipViolationImpliesSPR}
If $\tau.(\athread, \acom, \anact)\in\mmsem{\aprog}$ violates ownership, then it is an SPR.
\end{lemma}
The proof of 
Lemma~\ref{Lemma:OwnershipViolationImpliesSPR} (c.f. Appendix~\ref{Appendix:CheckSPRF})
is immediate from Lemma~\ref{Lemma:OwnImpliesLocal} and the definitions of ownership violation and strong pointer race.
The lemma implies the main result of this section: the memory-managed semantics coincides with the ownership-respecting semantics modulo SPRF (c.f. Appendix~\ref{Appendix:CheckSPRF}).
\begin{theorem}[SPRF Guarantee]
\label{Theorem:SPRFGuarantee}
If $\mmsem{\aprog}$ is SPRF, then $\mmsem{\aprog}=\resmmsem{\aprog}$.
\end{theorem}


\subsection{Checking SPRF}\label{Section:CheckSPRF}

This section establishes that checking SPRF may be done in the ownership-respecting semantics.
In other words, if $\mmsem{\aprog}$ has an SPR, then there is also one in $\resmmsem{\aprog}$.
This result, perhaps much less intuitively expected
than the symmetrical result of Section~\ref{Section:CheckPRF},
is particularly useful for optimizing thread-modular analysis of lock-free programs (cf. Section~\ref{Section:Evaluation}).
Its proof depends on a subtle interplay of ownership and validity. 

Let $\ownpof{\tau}$ be the \emph{owning pointers}, pointers in $\heapcomput{\tau}$ to addresses that are owned by threads and the next fields of addresses owned by threads.
To be included in $\ownpof{\tau}$, the pointers have to be valid.
A set of pointers $O\subseteq \ownpof{\tau}$ is \emph{coherent} if for all $\apexp,\apexpp\in\ownpof{\tau}$ with the same target or source address (in case of $\psel{\anadr}$ or $\dsel{\anadr}$) we have $\apexp\in O$ if and only if $\apexpp\in O$.

Lemma~\ref{Lemma:FreshOwn} below establishes the following fact.
For every computation that respects ownership, there is another one that coincides with it but assigns fresh cells to some of the owning pointers. 
To be more precise, given a computation $\tau\in\resmmsem{\aprog}$ and a coherent set of owning pointers $O\subseteq \ownpof{\tau}$, we can find another computation $\tau'\in\resmmsem{\aprog}$ where the resulting heap coincides with $\heapcomput{\tau}$ except for $O$.
These pointers are assigned fresh addresses. 
The proof of Lemma~\ref{Lemma:FreshOwn} is nontrivial and can be found in~Appendix~\ref{Appendix:CheckSPRF}.

\begin{lemma}\label{Lemma:FreshOwn}
Consider $\tau\in\resmmsem{\aprog}$ SPRF and $O\subseteq \ownpof{\tau}$ a coherent set.
There is $\tau'\in\resmmsem{\aprog}$ and an address mapping $\funa:\adrof{O}\rightarrow \adr$ that satisfy the following:
\begin{alignat*}{5}
(1)&&\hspace{0.1cm} \project{\tau}{\threads\times\coms}\ &=\ \project{\tau'}{\threads\times\coms}&\hspace{0.9cm}\freedof{\tau}&\subseteq\freedof{\tau'}&(4)&\\
(2)&&\hspace{0.1cm}\restrict{\heapcomput{\tau}}{\pexp\setminus O}\ &=\ \restrict{\heapcomput{\tau'}}{\pexp\setminus \fune(O)}&\ownpof{\tau'}&=(\ownpof{\tau}\setminus O)\ \cup\ \fune(O)&\hspace{0.1cm}(5)&\\
(3)&&\hspace{0.1cm}\restrict{\heapcomput{\tau}}{\validof{\tau}}\ &\heapiso\ \restrict{\heapcomput{\tau'}}{\validof{\tau'}}\hspace{0.2cm}\text{by}\hspace{0.2cm}\funa\cup \identity&
\adrof{\heapcomput{\tau}}&\cap\heapcomputof{\tau'}{\fune(O)}=\emptyset.&(6)&
\end{alignat*}
\end{lemma}
In this lemma, function $\funa$ specifies the new addresses that $\tau'$ assigns to the owning expressions in $O$. These new addresses are fresh by Point (6).
Point (1) says that $\tau$ and $\tau'$ are the same up to the particular addresses they manipulate, and 
Point (2) says that the reached states $\heapcomput\tau$ and $\heapcomput{\tau'}$ are the same up to the pointers touched by $\funa$.
Point~(3) states that the valid pointers of $\heapcomput\tau$ stay valid or become valid $\fune$-images of the originals. 
Point~(5) says that also the owned pointers of $\heapcomput\tau$ remain the same or become $\fune$-images of the originals. 
Finally, Point (4) says that $\heapcomput{\tau'}$ re-allocates less cells.

Lemma~\ref{Lemma:FreshOwn} is a cornerstone in the proof of the main result in this section, namely that SPRF is equivalent for the memory-managed and the ownership-respecting semantics.
\begin{theorem}[Checking SPRF]\label{Theorem:CheckSPRF}
$\mmsem{\aprog}$ is SPRF if and only if $\resmmsem{\aprog}$ is SPRF.
\end{theorem}
\begin{proof}
If $\mmsem{\aprog}$ is SPRF, by $\resmmsem{\aprog}\subseteq \mmsem{\aprog}$ this carries over to the ownership-respecting semantics.  
For the reverse direction, assume $\mmsem{\aprog}$ has an SPR.
In this case, there is a shortest computation $\tau.\anact\in\mmsem{\aprog}$ where $\anact$ raises an SPR.
In case $\tau.\anact\in\resmmsem{\aprog}$, we obtain the same SPR in the ownership-respecting semantics.

Assume $\tau.\anact\notin\resmmsem{\aprog}$.
We first argue that $\anact$ violates ownership.
By prefix closure, $\tau\in\mmsem{\aprog}$.
By minimality, $\tau$ is SPRF.
Since ownership violations are SPR by Lemma~\ref{Lemma:OwnershipViolationImpliesSPR}, $\tau$ does not contain any, $\tau\in \resmmsem{\aprog}$.
Hence, if $\anact$ respected ownership we could extend $\tau$ to the computation $\tau.\anact\in\resmmsem{\aprog}$ --- a contradiction to our assumption.

We turn this ownership violation in the memory-managed semantics into an SPR in the ownership-respecting semantics.
To this end, we construct a new computation $\tau'.\anact'\in\resmmsem{\aprog}$ that mimics $\tau.\anact$, respects ownership, but suffers from SPR.
Since $\tau.\anact$ is an ownership violation, $\anact$ takes the form $(\athread, \acom, \anup)$ with $\acom$ being
\begin{align*}
\psel{\apvarp}:=\apvar,\quad \dsel{\apvarp}:=\advar,\quad \text{or}\quad \freeof{\apvarp}.
\end{align*}
Here, $\heapcomputof{\tau}{\apvarp}\in\ownedof{\athread'}{\tau}$ and ($\athread'\neq \athread$ or $\apvarp\in\shared$).
Since the address is owned, Lemma~\ref{Lemma:OwnImpliesLocal} implies $\apvarp\notin\validof{\tau}$.

As a first step towards the new computation, we construct $\tau'$. 
Let $O:=\ownpof{\tau}$ be the (coherent) set of all owning pointers in all threads (with $\apvarp\notin O$). 
With this choice of $O$, we apply Lemma~\ref{Lemma:FreshOwn}.
It returns $\tau'\in\resmmsem{\aprog}$ with $\project{\tau'}{\threads\times\coms}\ =\ \project{\tau}{\threads\times\coms}$ and
\begin{align*}
	\restrict{\heapcomput{\tau'}}{\pexp\setminus\fune(O)}\ =\ \restrict{\heapcomput{\tau}}{\pexp\setminus O}
	\quad\text{and}\quad
	\restrict{\heapcomput{\tau'}}{\validof{\tau'}}\ \heapiso\ \restrict{\heapcomput{\tau}}{\validof{\tau}}.
\end{align*}
Address $\heapcomputof{\tau'}{\apvarp}$ is not owned by any thread.
This follows from $$\ownpof{\tau'}=(\ownpof{\tau}\setminus O)\ \cup\ \fune(O) = \fune(O)$$ and $q\not\in\fune(O)$.
Finally, $\apvarp\notin\validof{\tau'}$
by the isomorphism $\restrict{\heapcomput{\tau'}}{\validof{\tau'}}\ \heapiso\ \restrict{\heapcomput{\tau}}{\validof{\tau}}
$.

As a last step, we mimic $\anact=(\athread, \acom, \anup)$ by an action $\anact'=(\athread, \acom, \anup')$.
If $\acom$ is $\freeof{\apvarp}$, then we free the invalid pointer $\apvarp\notin\validof{\tau'}$ and obtain an SPR in $\resmmsem{\aprog}$. 
Assume $\acom$ is an assignment $\psel{\apvarp}:=\apvar$ (the case of $\dsel{\apvarp}:=\advar$ is similar).
Since $\anact$ is enabled after $\tau$ and $\heapcomputof{\tau'}{\apvarp}=\heapcomputof{\tau}{\apvarp}$, we have $\heapcomputof{\tau'}{\apvarp}\neq\segval$.
Hence, the command is also enabled after $\tau'$.
Since $\apvarp\notin\validof{\tau'}$, the assignment is again to an invalid pointer.
It is thus an SPR according by Definition~\ref{Definition:SPRF}.(i).
\qed
\end{proof}

\section{Improving Thread-Modular Analyses} 
\label{Section:Evaluation}
We now describe how the theory developed so far can be used to increase the efficiency of thread-modular analyses of pointer programs under explicit memory management. 

Thread-modular reasoning abstracts a program state into a set of states of individual threads. 
A thread's state consists of the local state, the part of the heap reachable from the local variables, and the shared state, the heap reachable from the shared variables. 

The analysis saturates the set of reachable thread states by a fixpoint computation.
Every step in this computation creates new thread states out of the existing ones by applying the following two rules.
\begin{inparaenum}[(1)]
	\item
		Sequential step: a thread's state is modified by an action of this thread.
	\item
		Interference: a state of a victim thread is changed by an action of another, interfering thread. 
		This is accounted for by creating combined two-threads states from  existing pairs of states of the victim and the interferer thread.
		The states that are combined have to agree on the shared part. 
		The combined state is constructed by deciding which addresses in the two local states coincide. 
		It is then observed how an action of the interferer changes the state of the victim within the combined state.
\end{inparaenum}

Pure thread-modular reasoning does not keep any information about what thread states can appear simultaneously during a computation and what identities can possibly hold between addresses of local states of threads.  
This brings efficiency, but also easily leads to false positives.
To see this, consider in Treiber's stack a state $s$ of a thread that is just about to perform the $\cas$ in \texttt{push}. 
Variable $\node$ points to an address $\anadr$ allocated in the first line of \texttt{push}, $\Topp$, $\topp$, and $\psel\node$ are at the top of the stack. 
Consider an interference step where the states $s_v$ of the victim and $s_i$ of the interferer are isomorphic to $s$, 
with $\node$ pointing to the newly allocated addresses $\anadr_v$ and $\anadr_i$, respectively.
Since the shared states conincide,  the interference is triggered.
The combination must account for all possible equalities among the local variables.
Hence, there is a combined state with $\anadr_v = \anadr_i$, which does not occur in reality.
This is a crucial imprecision, which leads to false positives.
Namely, the interferer's $\cas$ succeeds, resulting in the new victim's state $s_v'$ with $\Topp$ on $\anadr_i$ (which is equal to $\anadr_v$). 
The victim's $\cas$ then fails, and the thread continues with the commands $\topp:=\Topp;\psel\node:=\topp$. 
This results in $\psel{\anadr_v}$ pointing back to $\anadr_v$, and a loss of the stack content. 

Methods based on thread-modular reasoning must prevent such false positives by maintaining the necessary information about correlations of local states.
An efficient technique commonly used under garbage collection is based on ownership: 
a thread's state records that $\anadr$ has just been allocated and hence no other thread can access the address, until it enters the shared state.
This is enough to prevent false positives such as the one described above. 
Namely, the addresses $\anadr_i$ and $\anadr_v$ are owned by the respective threads and therefore they cannot be equal. 
Interference may then safely ignore the problematic case when $\anadr_v = \anadr_i$.
Moreover, besides the increased precision, the ability to avoid interference steps due to ownership significantly improves the overall efficiency.
This technique was used for instance to prove safety (and linearizability) of Treiber's stack and other subtle lock-free algorithms in~\cite{Vafeiadis:RGSep}.

Under explicit memory management, ownership of this form cannot be guaranteed. 
Addresses can be freed and re-allocated while still being pointed to.
Other techniques must be used to correlate the local states of threads.  
The solution chosen in~\cite{Sagiv:correlation,ThreadModular2013} is to replace the states of individual threads by states of pairs of threads.
Precision is thus restored at the cost of
an almost quadratic blow-up of the abstract domain
that in turn manifests itself in a severe decrease of scalability.
\subsection{Pointer Race Freedom Saves Ownership} 
\label{Section:PRFSavesOwnreship}
\label{par:improvements_due_to_pointer_race_freedom}

Using the results from Sections~\ref{Section:PRF} and \ref{Section:SPRF},
we show how to apply the ownership-based optimization of thread-modular reasoning to the memory-managed semantics.
To this end, we split the verification effort into two phases.
Depending on the notion of pointer race freedom,
we first check whether the program under scrutiny is (S)PRF. 
If the check fails, we report pointer races as potential errors to the developer. 
If the check succeeds, the second phase verifies the property of interest (here, linearizability) assuming (S)PRF.

When the notion of PRF from Section~\ref{Section:PRF} is used,
the second verification phase can be performed in the garbage-collected semantics due to Theorem~\ref{Theorem:PRFGuarantee}.
This allows us to apply the ownership-based optimization discussed above.
Moreover, Theorem~\ref{Theorem:CheckPRF} says that the first PR has to appear in the garbage-collected semantics. 
Hence, even the first phase, checking PRF, can rely on garbage collection and ownership.
The PRF check itself is simple. 
Validity of pointers is kept as a part of the individual thread states and updated at every sequential and interference step. 
Based on this, every computation step is checked for raising a PR according to Definition~\ref{Definition:PRF}. 
Our experiments suggest that the overhead caused by the recorded validity information is low.

For SPRF, we proceed analogously.
Due to the Theorems~\ref{Theorem:SPRFGuarantee} and~\ref{Theorem:CheckSPRF}, checking SPRF in the first phase and property verification in the second phase can both be done in the ownership-respecting semantics. 
The SPRF check is similar to the PRF check. 
Validity of pointers together with an information about strong invalidity is kept as a part of a thread's state, and every step is checked for raising an SPR according to Definition~\ref{Definition:SPRF}.

The surprising good news is that both phases can again use the ownership-based optimization. 
That is, also in the ownership-respecting semantics, interferences on the owned memory addresses can be skipped.
We argue that this is sound.
Due to Lemma~\ref{Lemma:OwnImpliesLocal}, 
if a thread $\athread$ owns an address $\anadr$, other threads may access $\anadr$  only via invalid pointers. 
Therefore,
\begin{inparaenum}[(1)]
	\item 
modifications of $\anadr$ by $\athread$ need not be considered as an interference step for other threads.
Indeed, if a thread $\athread'\neq \athread$ was influenced by such a modification ($\athread'$ reads a next or the data field of $\anadr$), 
then the corresponding variable of $\athread'$ would become strongly invalid, Definition~\ref{Definition:StrongInvalidity}. 
Hence, either this variable is never used in an assertion or in a dereference again (it is effectively dead), 
or the first use raises an SPR, Cases~(ii) and~(iii) in Definition~\ref{Definition:SPRF}.  
	\item 
In turn, in the ownership-respecting semantics, another thread $\athread'$ cannot make changes to $\anadr$, by Definition~\ref{Definition:OwnershipViolation} of ownership violations.  
This means we can also avoid the step where $\athread'$ interferes with the victim $\athread$. 
\end{inparaenum}

\subsection{Experimental Results} 
\label{Section:ExperimentalResults}

To substantiate our claim for a more efficient analysis with practical experiments, we implemented the thread-modular analysis from~\cite{ThreadModular2013} in a prototype tool.
This analysis is challenging for three reasons: it checks  linearizability as a non-trivial requirement, it handles an unbounded number of threads, and it supports an unbounded heap. 
Our tool covers the garbage-collected semantics, the new ownership-respecting semantics of Section~\ref{Section:SPRF}, and the memory-managed semantics.
For the former two, we use the abstract domain where local states refer to single threads.
Moreover, we support the ownership-based pruning of interference steps from Section~\ref{Section:PRFSavesOwnreship}. 
For the memory-managed semantics, to restore precision as discussed above, the abstract domain needs local states with pairs of threads. 
Rather than running two phases, our tool combines the PRF check and the actual analysis. 
We tested our implementation on lock-free data structures from the literature and verified linearizability following the approach in~\cite{ThreadModular2013}.

\begin{table}
	\caption{Experimental results for thread-modular reasoning using different memory semantics.}
	\label{tab:experiments}%
	\def\firstcolwidth{2.7cm}%
	\newcolumntype{Y}{>{\centering\arraybackslash}X}%
	\newcolumntype{Z}{>{\raggedright}m}%
	\vspace{-0.5cm}
	\begin{center}
	\scalebox{0.8}{
	\begin{tabularx}{1.2\textwidth}{Z{\firstcolwidth+.1cm}l*{5}{Y}Y}
		\toprule[0.1ex]
		\multicolumn{2}{l}{Program}
			& time in seconds
			& explored state count
			& sequential step count
			& interference step count
			& pruned interferences
			& correctness established
			\\
		\midrule[0.3ex]
		\multirow{5}{\firstcolwidth}{Single lock stack}
			& GC        & 0.053  & 328    & 941    & 3276    & 10160	& yes \\
			& OWN       & 0.21   & 703    & 1913   & 6983    & 22678	& yes \\
			& GC$^-$    & 0.20   & 507    & 1243   & 19321   & --   	& yes \\
			& OWN$^-$   & 0.60   & 950    & 2474   & 38117   & --   	& yes \\
			& MM$^-$    & 5.34   & 16117  & 25472  & 183388  & --   	& yes \\
		\hdashline[1pt/1pt]
		\multirow{5}{\firstcolwidth}{Single lock queue}
			& GC        & 0.034  & 199    & 588    & 738     & 5718 	& yes \\
			& OWN       & 0.56   & 520    & 1336   & 734     & 31200	& yes \\
			& GC$^-$    & 0.19   & 331    & 778    & 9539    & --   	& yes \\
			& OWN$^-$   & 2.52   & 790    & 1963   & 65025   & --   	& yes \\
			& MM$^-$    & 31.7   & 27499  & 60263  & 442306  & --   	& yes \\
		\hdashline[1pt/1pt]
		\multirow{5}{\firstcolwidth}{Treiber's lock free stack (with version counters) \cite{MichaelScottTreiber}}
			& GC        & 0.052  & 269     & 779     & 3516      & 15379	& yes\\
			& OWN       & 2.36   & 744     & 2637    & 43261     & 95398	& yes\\
			& GC$^-$    & 0.16   & 296     & 837     & 11530     & --   	& yes\\
			& OWN$^-$   & 4.21   & 746     & 2158    & 73478     & --   	& yes\\
			& MM$^-$    & 602    & 116776  & 322057  & 7920186   & --   	& yes\\
		\hdashline[1pt/1pt]
		\multirow{5}{\firstcolwidth}{Michael \& Scott's lock free queue \cite{MichaelScottTreiber} (with hints)}
			& GC        & 2.52             & 3134          & 6607          & 46838         & 1237012 		& yes\\
			& OWN       & 10564            & 19553         & 43305         & 6678240       & 20747559		& yes\\
			& GC$^-$    & 9.08             & 3309          & 7753          & 187349        & --      		& yes\\
			& OWN$^-$   & 51046            & 31329         & 64234         & 35477171      & --      		& yes\\
			& MM$^-$    & \textit{aborted} & $\ge\,$69000  & $\ge\,$90000  & --            & --      		& false positive \\
		\bottomrule[0.1ex]
	\end{tabularx}}
	\end{center}
\end{table}

The experimental results are listed in Table~\ref{tab:experiments}.
The experiments were conducted on an Intel Xeon E5-2650 v3 running at 2.3 GHz.
The table includes the following:
\begin{inparaenum}[(1)]
	\item runtime taken to establish correctness,
	\item number of explored thread states (i.e. size of the search space),
	\item number of sequential steps,
	\item number of interference steps,
	\item number of interference steps that have been pruned by the ownership-based optimization, and
	\item the result of the analysis, i.e. whether or not correctness could be established. 
\end{inparaenum}
For a comparison, we also include the results with the ownership-based optimization turned off (suffix $^-$). 
Recall that the optimization does not apply to the memory-managed semantics. 
We elaborate on our findings.

Our experiments confirm the usefulness of pointer race freedom.
When equipped with pruning (OWN), the ownership-respecting semantics provides a speed-up of two orders of magnitude for Treiber's stack and the single lock data structures compared to the memory-managed semantics (MM$^-$).
The size of the explored state space is close to the one for the garbage-collected semantics (GC) and up to two orders of magnitude smaller than the one for explicit memory management.
We also stress tested our tool by purposely inserting pointer races, for example, by discarding the version counters.
In all cases, the tool was able to detect those races.

For Michael \& Scott's queue we had to provide hints in order to eliminate certain patterns of false positives. 
This is due to an imprecision that results from joins over a large number of states (we are using the joined representation of states from \cite{ThreadModular2013} based on Cartesian abstraction).
Those hints are sufficient for the analysis relying on the ownership-respecting semantics to establish correctness.
The memory-manged semantics produces more false positives, the elimination of which would require more hinting, as also witnessed by the implementation of \cite{ThreadModular2013}.
Regarding the stress tests from above, note that we ran those tests with the same hints and were still able to find the purposely inserted bugs.
\vspace{-0.2cm}
\section{Conclusion} 
\label{Section:conclusion}
\vspace{-0.2cm}

We have conducted a semantic study on the relationship between concurrent heap-manipulating programs running under explicit memory management and under garbage collection.
We proposed the notion of pointer race that captures the difference between the two semantics and characterizes common synchronizations errors 
similar to the well-known data races.
We proved that the verification of pointer race free programs under explicit memory management can be reduced to the easier verification under garbage collection. 
We showed an analogous result with a stronger notion of pointer race proposed to fit performance critical (e.g. lock-free) implementations, which are intentionally racy in our original sense.
Our results are particularly useful in thread-modular analysis under explicit memory-management.
We showed that they allow us to apply an ownership-based optimization available before only under garbage-collection.
Using this optimization, our prototype was able to verify lock-free algorithms like Treiber's stack and Micheal \& Scott's queue for the memory-managed semantics with a performance gain of up to two orders of magnitude.

\bibliographystyle{plain}
\bibliography{cited}

\newpage
\appendix
\section{Missing Details}
The intended behavior of Treiber's stack is as follows, see Figure~\ref{Example:Treiber}.
Upon a push, the corresponding thread allocates a new cell using a local pointer variable $\node$ and sets the given value.
In a loop, the thread now tries to alter the top of stack. 
It sets a local pointer variable $\topp$ to the old top of stack stored in the global pointer variable $\Topp$.
Then it redirects the next selector of $\node$ to the old top of stack.
If no concurrent execution of a push or a pop has interefered, $\Topp$ and $\topp$ still point to the same cell and the thread atomically sets $\Topp$ to $\node$.
To be precise, the compare-and-swap command $\cas$ atomically checks the equality $\Topp=\topp$ and, in case it holds, assigns to $\Topp$ the value of $\node$ and returns true.
If the values differ, the command returns false. 
We decided not to add $\cas$ to the set of commands to keep our instruction set small.
The theory can be extended to cover $\cas$.

The pop method also creates a local copy $\topp$ of the global top of stack.
It checks whether the stack is empty and, in case, returns negatively.
Otherwise, the method copies the new top of stack $\psel{\topp}$ into the local variable $\node$ and atomically moves the global top of stack $\Topp$ to $\node$.
Now the thread executing pop can access the value of the cell, free $\topp$, and return.

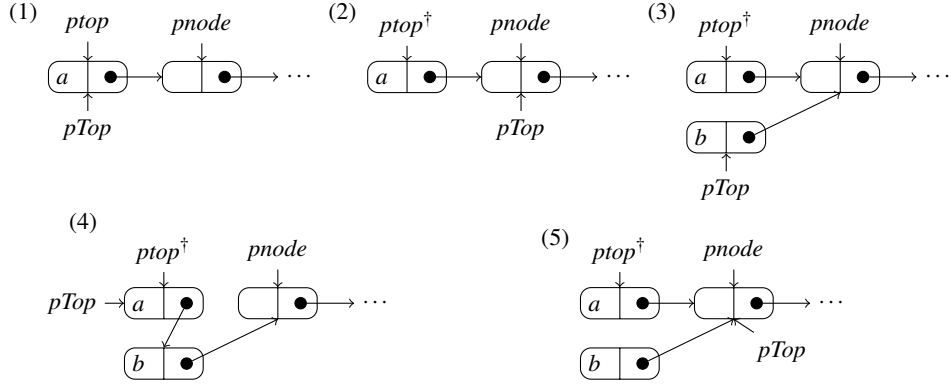
\begin{figure}
\begin{tikzpicture}[scale=2]
    \tikzstyle{cell} = [rectangle split,rectangle split horizontal, rectangle split parts=2,draw,rounded corners,text width=0.3cm,text height=0.2cm]
    \tikzstyle{next} = [circle,minimum size=0.15cm,inner sep=0pt,fill=black,draw=black]
    \node[cell] (a){$a$\nodepart{two}};
    \node (nr)[above left of=a, node distance=1.2cm]{(1)};
    \node[next](apointer)[right of=a, node distance=0.3cm] {};
    \node (top) [above of=a,node distance=0.7cm]{$\topp$};
    \node (Top) [below of=a,node distance=0.7cm]{$\Topp$};
    \draw[->](top)--(a.north);
    \draw[->](Top)--(a.south);
    \node[cell] (other)[right of=a, node distance=1.5cm]{\nodepart{two}};
    \draw[->](apointer)--(other.west);
    \node[next](otherpointer)[right of=other, node distance=0.3cm] {};
    \node (node) [above of=other,node distance=0.7cm]{$\node$};
    \draw[->](node)--(other.north);
    \node[](dots)[right of=other, node distance=1.3cm] {$\cdots$};
    \draw[->](otherpointer)--(dots.west);
\begin{scope}[xshift=2.1cm]
    \node[cell] (a){$a$\nodepart{two}};
    \node (nr)[above left of=a, node distance=1.2cm]{(2)};
    \node[next](apointer)[right of=a, node distance=0.3cm] {};
    \node (top) [above of=a,node distance=0.7cm]{$\topp^{\dagger}$};
    \draw[->](top)--(a.north);
    \node[cell] (other)[right of=a, node distance=1.5cm]{\nodepart{two}};
    \draw[->](apointer)--(other.west);
    \node (Top) [below of=other,node distance=0.7cm]{$\Topp$};
    \draw[->](Top)--(other.south);
    \node[next](otherpointer)[right of=other, node distance=0.3cm] {};
    \node (node) [above of=other,node distance=0.7cm]{$\node$};
    \draw[->](node)--(other.north);
    \node[](dots)[right of=other, node distance=1.3cm] {$\cdots$};
    \draw[->](otherpointer)--(dots.west);
\end{scope}
\begin{scope}[xshift=4.2cm]
    \node[cell] (a){$a$\nodepart{two}};
    \node (nr)[above left of=a, node distance=1.2cm]{(3)};
    \node[next](apointer)[right of=a, node distance=0.3cm] {};
    \node (top) [above of=a,node distance=0.7cm]{$\topp^{\dagger}$};
    \draw[->](top)--(a.north);
    \node[cell] (b)[below of=a, node distance=0.8cm]{$b$\nodepart{two}};
    \node[next](bpointer)[right of=b, node distance=0.3cm] {};
    \node[cell] (other)[right of=a, node distance=1.5cm]{\nodepart{two}};
    \draw[->](apointer)--(other.west);
    \draw[->](bpointer)--(other.south);
    \node (Top) [below of=b,node distance=0.7cm]{$\Topp$};
    \draw[->](Top)--(b.south);
    \node[next](otherpointer)[right of=other, node distance=0.3cm] {};
    \node (node) [above of=other,node distance=0.7cm]{$\node$};
    \draw[->](node)--(other.north);
    \node[](dots)[right of=other, node distance=1.3cm] {$\cdots$};
    \draw[->](otherpointer)--(dots.west);
\end{scope}
\begin{scope}[yshift=-1.5cm,xshift=0.5cm]
    \node[cell] (a){$a$\nodepart{two}};
    \node (nr)[above left of=a, node distance=1.5cm]{(4)};
    \node[next](apointer)[right of=a, node distance=0.3cm] {};
    \node (top) [above of=a,node distance=0.7cm]{$\topp^{\dagger}$};
    \draw[->](top)--(a.north);
    \node[cell] (b)[below of=a, node distance=0.8cm]{$b$\nodepart{two}};
    \node[next](bpointer)[right of=b, node distance=0.3cm] {};
    \node (Top) [left of=a,node distance=1.2cm]{$\Topp$};
    \draw[->](Top)--(a.west);
    \node[cell] (other)[right of=a, node distance=1.5cm]{\nodepart{two}};
    \draw[->](apointer)--(b.north);
    \draw[->](bpointer)--(other.south);
    \node[next](otherpointer)[right of=other, node distance=0.3cm] {};
    \node (node) [above of=other,node distance=0.7cm]{$\node$};
    \draw[->](node)--(other.north);
    \node[](dots)[right of=other, node distance=1.3cm] {$\cdots$};
    \draw[->](otherpointer)--(dots.west);
\end{scope}
\begin{scope}[yshift=-1.5cm,xshift=3.5cm]
    \node[cell] (a){$a$\nodepart{two}};
    \node (nr)[above left of=a, node distance=1.2cm]{(5)};
    \node[next](apointer)[right of=a, node distance=0.3cm] {};
    \node (top) [above of=a,node distance=0.7cm]{$\topp^{\dagger}$};
    \draw[->](top)--(a.north);
    \node[cell] (b)[below of=a, node distance=0.8cm]{$b$\nodepart{two}};
    \node[next](bpointer)[right of=b, node distance=0.3cm] {};
    \node[cell] (other)[right of=a, node distance=1.5cm]{\nodepart{two}};
    \draw[->](apointer)--(other.west);
    \draw[->](bpointer)--(other.south);
    \node[next](otherpointer)[right of=other, node distance=0.3cm] {};
    \node (Top) [below right of=other,node distance=0.9cm]{$\Topp$};
    \draw[->](Top)--(other.south);
    \node (node) [above of=other,node distance=0.7cm]{$\node$};
    \draw[->](node)--(other.north);
    \node[](dots)[right of=other, node distance=1.3cm] {$\cdots$};
    \draw[->](otherpointer)--(dots.west);
\end{scope}
\end{tikzpicture}
\caption{ABA-problem in Treiber's stack.}
\label{Figure:ABA}
\end{figure}


\section{Proofs in Section~\ref{Section:PRF}}
Throughout the appendix, we refer to the two equations for heap isomorphism as \emph{compatibility requirements}.
\begin{proof}[of Lemma~\ref{Lemma:HeapIsoAlgebra}]
We consider the first claim.
Let $\aheap_1\heapiso\aheap_2$ via $\anisoa:\adrof{\aheap_1}\rightarrow \adrof{\aheap_2}$.
Let $A:=\adrof{\restrict{\aheap_1}{P}}$.
The task is to show that $\anisoa':=\restrict{\anisoa}{A}:A\rightarrow\anisoa(A)$
defines an isomorphism between
$\restrict{\aheap_1}{P}$ and $\restrict{\aheap_2}{\anisoe(P)}$.
To this end, it is sufficient to show that $\anisoa'$ induces a bijection $\anisoe'$ between $\domof{\restrict{\aheap_1}{P}}$ and $\domof{\restrict{\aheap_2}{\anisoe(P)}}$ that satisfies the compatibility requirements for an isomorphism.
From this we derive that $\anisoa'$ is a bijection between the addresses.

Function $\anisoe'$ is total since $\anisoa'$ maps all addresses in 
$\domof{\restrict{\apval_1}{P}}\cap \adr$ and in $\domof{\restrict{\adval_1}{D}}\cap \adr$.
The function is injective as $\anisoe$ is.
In the case of pointer expressions, surjectivity means for every $\apexpp\in \domof{\restrict{\apval_2}{\anisoe(P)}}$ there is $\apexp\in \domof{\restrict{\apval_1}{P}}$ with $\anisoe(\apexp)=\apexpp$.
Since $\apexpp\in \domof{\restrict{\apval_2}{\anisoe(P)}}$, we have $\apexpp\in \domof{\apval_2}$. 
Since $\anisoa$ is a heap isomorphism, there is $\apexp\in \domof{\apval_1}$ with $\anisoe(\apexp)=\apexpp$.
Moreover, since $\apexpp\in \anisoe(P)$, we have $\apexp\in P$.
Together, $\apexp\in \domof{\restrict{\apval_1}{P}}$.
Surjectivity for data expressions is similar.
The compatibility required for an isomorphism holds as it holds for $\anisoe$.
\\[0.2cm]
Consider now the second claim and assume $\anadr\in\adrof{\aheap_1}$ but $\anadrp\notin\adrof{\aheap_1}$.
Let $\anisoa$ be the isomorphism between $\aheap_1$ and $\aheap_2$.
We extend the function by $\anadrp\mapsto\anadrp'$ and restrict it to the new domain $\adrof{\aheap_1[\psel{\anadr}\mapsto\anadrp]}$.
Note that we indeed may lose the address that $\psel{\anadr}$ was pointing to
so that a restriction is necessary:
\begin{align*}
\anisoa':=\restrict{(\anisoa\cup\set{\anadrp\mapsto\anadrp'})}{\adrof{\aheap_1[\psel{\anadr}\mapsto\anadrp]}}.
\end{align*}
We first check that $\anisoa'$ is a function.
This holds as $\anadrp\notin\adrof{\aheap_1}=\domof{\anisoa}$.
To show that $\anisoa'$ is a bijection between the addresses, it is again sufficient to show that the induced function on pointer and data expressions $\anisoe'$ is a bijection satisfying the requirements for an isomorphism.
The induced function is total as $\anisoa'$ maps all addresses in 
$\adrof{\aheap_1[\psel{\anadr}\mapsto\anadrp]}$.
In particular does $\anisoa'$ extend $\anisoa$ and hence map $\adrof{\aheap_1}$
containing $\anadr$.
The induced function is injective essentially as $\anisoe$ is. 
To be precise, if $\psel{\anadr}$ was not in the domain of $\aheap_1$, then $\psel{\anisoa(\anadr)}$ was not in the domain of $\aheap_2$ since $\anisoe$ is a bijection. 
We can therefore safely add this mapping.
It remains to show that the function is surjective. 
Consider $\apexpp\in \domof{\aheap_2[\psel{\anadr'}\mapsto\anadrp']}$ with $\anadr'=\anisoa(\anadr)$.
If $\apexpp=\psel{\anadr'}$, then $\apexpp=\anisoe'(\psel{\anadr})$.
If we have $\apexpp\in \domof{\aheap_2}\setminus\set{\psel{\anadr'}}$, then
there is $\apexp\in \domof{\aheap_1}$ so that $\apexpp=\anisoe(\apexp)$.
This $\apexp$ exists as $\anisoe$ is a bijection between the old domains.

To check the compatibility requirements, let $\apval_1$ be the pointer valuation in $\aheap_1[\psel{\anadr}\mapsto\anadrp]$ and let $\apval_2$ be the valuation in $\aheap_2[\psel{\anadr'}\mapsto\anadrp']$. Then
\begin{align*}
\anisoa'(\apval_1(\psel{\anadr})) &= \anisoa'(\anadrp) \\
&= \anadrp'\\
&= \apval_2(\psel{\anadr'}) 
= \apval_2(\anisoe'(\psel{\anadr})).
\end{align*}
For the remaining pointers, the requirement holds by the fact that $\anisoa$ was a heap isomorphism. 
\qed
\end{proof}
\begin{proof}[of Lemma~\ref{Lemma:Freed}]
Assume $\anadr\in \freedof{\tau}$.
To show that $\anadr\notin \adrof{\restrict{\heapcomput{\tau}}{\validof{\tau}}}$, 
we show that (i) no pointer $\apexp$ to $\anadr$ is valid (in $\validof{\tau}$) and (ii) no selector $\psel{\anadr}$ is valid. 
Since the restriction $\restrict{\heapcomput{\tau}}{\validof{\tau}}$ only keeps selectors $\dsel{\anadr}$ for valid pointers to $\anadr$, Argument~(i) also removes $\anadr$ from $\domof{\adval_{\tau}}$.

For (i), we show that $\heapcomputof{\tau}{\apexp}=\anadr$ implies $\apexp\notin\validof{\tau}$.
Since $\anadr\in\freedof{\tau}$, there was no malloc after the free of the address.
Hence, there are two cases.
Either $\apexp$ learned about $\anadr$ before or after the address was freed.
In the former case, $\apexp$ is invalidated by the free of address $\anadr$.
In the latter case, $\apexp$ can only learn about $\anadr$ from an invalid pointer.
This renders $\apexp$ invalid, too.

We show (ii), all selectors $\psel{\anadr}$ are invalid.
These selectors were declared invalid at the moment address $\anadr$ was freed.
The only way to validate $\psel{\anadr}$ is via an assignment to it.
This assignment is forbidden as computation $\tau$ is assumed to be PRF.
Indeed, with Argument~(i), all pointers to $\anadr$ are invalid and hence accessing the next selector will result in a pointer race according to Definition~\ref{Definition:PRF}(i).
\qed
\end{proof}
\begin{proof}[of Lemma~\ref{Lemma:GCPRF}]
We proceed by induction.
In the base case, all pointers are valid and there is nothing to prove.
Assume the claim holds for $\sigma$ and consider $\sigma.\anact\in\gcsem{\aprog}$ PRF.
The main task is to carefully consider the assignments.
There are five cases. 
\begin{description}

\item[Case $\psel{\apvar}:=\apvarp$] Let $\heapcomputof{\sigma}{\apvar}=\anadr\neq \segval$. We have $\anadr\neq \segval$ by enabledness. 
\begin{itemize}
\item[(1)] If $\apvarp$ is invalid, it points to $\anadrp\neq \segval$ by the induction hypothesis(i). If $\apvar$ was valid, claim(ii) now holds for $\psel{\anadr}$.
\item[(2)] If $\apvarp$ is valid pointing to $\anadrp$ which may be $\segval$, then $\psel{\anadr}$ will be valid. 
If $\anadrp\neq \segval$, the next selectors of $\anadrp$ behave as required by the induction hypothesis(ii) on $\apvarp$.
If $\anadrp=\segval$, claim(ii) is trivial for $\psel{\anadr}$.
\end{itemize}



\item[Case $\apvar:=\apvarp$] $\phantom{Text}$
\begin{itemize}
\item[(3)]
If $\apvarp$ is invalid, it points to $\anadr\neq \segval$ by the induction hypothesis(i). 
This proves claim(i) for $\apvar$.
\item[(4)]
If $\apvarp$ is valid, then $\apvar$ will become valid.
If $\heapcomputof{\sigma}{\apvarp}=\segval$, claim(ii) trivially holds for $\apvar$.
If $\heapcomputof{\sigma}{\apvarp}=\anadrp\neq \segval$, the statement about the invalid next selectors of $\anadrp$ carries over to $\apvar$ by the induction hypothesis(ii) on $\apvarp$.
\end{itemize}

\item[Case $\apvar:=\psel{\apvarp}$]
Let $\heapcomputof{\sigma}{\apvarp}=\anadrp\neq \segval$. Again, we have $\anadrp\neq \segval$ by enabledness.
\begin{itemize}
\item[(5)]
If $\apvarp$ is invalid, this is again a pointer race. 
\item[(6)]
This means $\apvarp$ is valid.
If $\heapcomputof{\sigma}{\psel{\anadrp}}=\segval$, by the induction hypothesis(ii) expression $\psel{\anadrp}$ has to be valid.
Then $\apvar$ becomes valid and points to $\segval$.
In this case, claim(ii) trivially holds.
Otherwise, $\heapcomputof{\sigma}{\psel{\anadrp}}=\anadrpp\neq \segval$. 
If $\psel{\anadrp}$ is invalid, $\apvar$ becomes invalid and claim(i) holds.
If $\psel{\anadrp}$ is valid, claim(ii) about the next selectors of $\anadrpp$ holds for $\apvar$ by the induction hypothesis(ii) on $\psel{\anadrp}$.
\end{itemize}
\end{description}
Consider a free command $\freeof{\apvar}$ with $\heapcomputof{\sigma}{\apvar}=\anadr\neq \segval$. 
This invalidates all pointers to $\anadr$ and claim(i) and claim(ii) hold.

Consider a malloc $\apvar:=\malloc$.
This returns a fresh cell $f$ where all next selectors $\psel{f}$ are valid.
Hence, claim(ii) trivially holds.\qed
\end{proof}
\begin{proof}[of Proposition~\ref{Proposition:PRFImpliesGC}]
We proceed by induction on the length of $\tau$.
The base case of single actions setting data variables to arbitrary values is trivial.
In the induction step, assume for $\tau$ we have the heap-equivalent computation~$\sigma$. 
The fact that the sequences of commands coincide for $\tau$ and $\sigma$ means we can assume the resulting control states to coincide.
This allows us to always choose the same next command in both semantics.
We therefore focus on the heap content and do a case distinction along the transition rule that leads to $\tau'$. 
In the following, let $\anisoa:\adrof{\restrict{\heapcomput{\tau}}{\validof{\tau}}}\rightarrow \adrof{\restrict{\heapcomput{\sigma}}{\validof{\sigma}}}$ be the isomorphism between $\restrict{\heapcomput{\tau}}{\validof{\tau}}$ and $\restrict{\heapcomput{\sigma}}{\validof{\sigma}}$.\\[0.2cm]
\bfemph{Case (Malloc2)}\quad
Let $\tau'=\tau.(\athread, \apvar:=\malloc, [\apvar\mapsto \anadr])$ with $\anadr\in \freedof{\tau}$. 
It can be shown that the value of $\dsel{\anadr}$ is defined for addresses $\anadr$ that have been freed.
Let it be $\heapcomputof{\tau}{\dsel{\anadr}}=d$.
We now have
\begin{align}
\restrict{\heapcomput{\tau'}}{\validof{\tau'}}
=&\ \restrict{\heapcomput{\tau}[\apvar\mapsto \anadr]}{\validof{\tau}\cup\set{\apvar}} \notag\\
=&\ (\restrict{\heapcomput{\tau}}{\validof{\tau}})[\apvar\mapsto \anadr,\dsel{\anadr}\mapsto d].\label{Equation:Malloc2FormTau}
\end{align}
The first equation is by the definition of $\tau'$ and $\validof{\tau'}$.
To understand the second equation, note that $\anadr\notin\adrof{\restrict{\heapcomput{\tau}}{\validof{\tau}}}$ by $\anadr\in\freedof{\tau}$ and Lemma~\ref{Lemma:Freed}.
This means in
$\restrict{\heapcomput{\tau}[\apvar\mapsto \anadr]}{\validof{\tau}\cup\set{\apvar}}$, pointer $\apvar$ is the only reference to $\anadr$.
Since we have a reference to $\anadr$, value $\dsel{\anadr}$ is defined in $\restrict{\heapcomput{\tau}[\apvar\mapsto \anadr]}{\validof{\tau}\cup\set{\apvar}}$.
So if we push the restriction over the update, we have to preserve definedness of $\dsel{\anadr}$.
Therefore, we add $\dsel{\anadr}\mapsto d$ to the update.

To mimic the command with garbage collection, we apply 
Rule~(Malloc1) and get $\sigma':=\sigma.(\athread, \apvar:=\malloc, [\apvar\mapsto f, \dsel{f}\mapsto d, \setcond{\psel{f}\mapsto\segval}{\text{for every selector $\texttt{next}$}}])$.
Note that we allocate a fresh address $f\notin \adrof{\heapcomput{\sigma}}$.
The transition rule allows us to select an arbitrary value. 
We choose the value of $\dsel{\anadr}$ in $\heapcomput{\tau}$.
The next selectors are yet undefined.
We have 
\begin{align}
\restrict{\heapcomput{\sigma'}}{\validof{\sigma'}} 
=&\ \restrict{\heapcomput{\sigma}[\apvar\mapsto f, \dsel{f}\mapsto d, \setcond{\psel{f}\mapsto\segval}{\text{for every $\texttt{next}$}}]}{\validof{\sigma}\cup\set{\apvar}}\notag\\
=&\ (\restrict{\heapcomput{\sigma}}{\validof{\sigma}})[\apvar\mapsto f, \dsel{f}\mapsto d].\label{Equation:Malloc2FormSigma}
\end{align}
The first equation is again by the definition of $\sigma'$ and of $\validof{\sigma'}$.
The second equation preserves the assignments $\apvar\mapsto f$ and $\dsel{f}\mapsto\advalue$.
Since we overwrite the value of $\apvar$, there is no need to keep $\apvar$ with the valid pointers when we push the restriction inside.
Since $f$ is fresh, the pointers $\psel{f}$ are not contained in $\validof{\sigma}$. 
The restriction removes the corresponding assignments to $\segval$.

To see that $\restrict{\heapcomput{\tau'}}{\validof{\tau'}}$ and $\restrict{\heapcomput{\sigma'}}{\validof{\sigma'}}$ are isomorphic, we first note that
\begin{align*}
\restrict{\heapcomput{\tau}}{\validof{\tau}}\heapiso
\restrict{\heapcomput{\sigma}}{\validof{\sigma}}
\end{align*}
by the induction hypothesis.
We already argued for $\anadr\notin\adrof{\restrict{\heapcomput{\tau}}{\validof{\tau}}}$. Similarly, $f\notin\adrof{\restrict{\heapcomput{\sigma}}{\validof{\sigma}}}$.
This allows us to apply Lemma~\ref{Lemma:HeapIsoAlgebra}\eqref{Equation:HeapIsoModifyPointer}, more precisely a variant of Case~\eqref{Equation:HeapIsoModifyPointer} where the next pointer is replaced by $\apvar$:
\begin{align*}
(\restrict{\heapcomput{\tau}}{\validof{\tau}})[\apvar\mapsto\anadr]\heapiso
(\restrict{\heapcomput{\sigma}}{\validof{\sigma}})[\apvar\mapsto f].
\end{align*}
The isomorphism for this new heap maps $\anadr$ to $f$.
This allows us to apply Lemma~\ref{Lemma:HeapIsoAlgebra}\eqref{Equation:HeapIsoModifyData} and get
\begin{align*}
(\restrict{\heapcomput{\tau}}{\validof{\tau}})[\apvar\mapsto\anadr][\dsel{\anadr}\mapsto d]\heapiso
(\restrict{\heapcomput{\sigma}}{\validof{\sigma}})[\apvar\mapsto f][\dsel{f}\mapsto d].
\end{align*}
With Equations~\eqref{Equation:Malloc2FormTau} and~\eqref{Equation:Malloc2FormSigma}, this is the desired
\begin{align*}
\restrict{\heapcomput{\tau'}}{\validof{\tau'}}\heapiso \restrict{\heapcomput{\sigma'}}{\validof{\sigma'}}.
\end{align*}
\bfemph{Case (Free)}\quad
Let $\tau'=\tau.(\athread, \freeof{\apvar}, \emptyset)$.
Since $\tau'$ is assumed to be PRF, we get $\apvar\in \validof{\tau}$.
The more complex case is that $\heapcomputof{\tau}{\apvar}=\anadr\neq \segval$. 
The pointers that remain valid after the free are
\begin{align*}
\validof{\tau'}&=\validof{\tau}\setminus \invalidof{\anadr}.
\end{align*}
Recall that $\invalidof{\anadr}:=\setcond{\apexp}{\heapcomputof{\tau}{\apexp}=\anadr}\cup \set{\pselarg{\anadr}{1},\ldots, \pselarg{\anadr}{n}}$.
Then
\begin{align}
\restrict{\heapcomput{\tau'}}{\validof{\tau'}} 
=&\ \restrict{\heapcomput{\tau}}{\validof{\tau}\setminus \invalidof{\anadr}}  \label{Equation:FreeFormTau}
\end{align}

To mimic the command with garbage collection, we also free pointer $\apvar$ in $\sigma$ and obtain $\sigma':=\sigma.(\athread, \freeof{\apvar}, \emptyset)$.
Since $\apvar$ is defined in $\restrict{\heapcomput{\tau}}{\validof{\tau}}$, heap isomorphism requires
$\apval_{\sigma}(\apvar)=\anisoa(\apval_{\tau}(\apvar))=\anisoa(\anadr)\neq \segval$. We have $\anisoa(\anadr)\neq\segval$ as $\anadr\neq \segval$ and $\anisoa$ is an address mapping.

The pointers that remain valid after the free are
\begin{align*}
\validof{\sigma'}&= \validof{\sigma}\setminus \invalidof{\anisoa(\anadr)}.
\end{align*}
As in the case of $\tau'$, we obtain
\begin{align}
\restrict{\heapcomput{\sigma'}}{\validof{\sigma'}} 
= \restrict{\heapcomput{\sigma}}{\validof{\sigma}\setminus \invalidof{\anisoa(\anadr)}}. \label{Equation:FreeFormSigma}
\end{align}

For the isomorphism, we first show that 
\begin{align}
\anisoe(\validof{\tau}\setminus{\invalidof{\anadr}})=\validof{\sigma}\setminus \invalidof{\anisoa(\anadr)}. \label{Equation:IsoInvalid}
\end{align}
To prove the inclusion from left to right, consider $\anisoe(\apexp)$ with $\apexp\in \validof{\tau}$ and $\apexp\notin \invalidof{\anadr}$.
Since $\apexp\in \validof{\tau}$, we have that $\anisoe(\apexp)\in\validof{\sigma}$.
This holds since $\anisoe$ defines a bijection between $\domof{\restrict{\heapcomput{\tau}}{\validof{\tau}}}$ and $\domof{\restrict{\heapcomput{\sigma}}{\validof{\sigma}}}$. 
To see that $\anisoe(\apexp)\notin \invalidof{\anisoa(\anadr)}$, assume for the sake of contradiction that it was in the set.
This either means $\anisoe(\apexp)$ points to $\anisoa(\anadr)$ or it is a selector of $\anisoa(\anadr)$.
Consider the former case.
Then we have
\begin{align*}
\anisoa(\anadr)=\heapcomputof{\sigma}{\anisoe(\apexp)} = \anisoa(\heapcomputof{\tau}{\apexp}). 
\end{align*}
The second equation holds by the fact that $\anisoa$ is a heap isomorphism.
Together, we get $
\heapcomputof{\tau}{\apexp}= \anadr$. 
This contradicts the fact that $\apexp\notin\invalidof{\anadr}$.
The reverse inclusion is along similar lines.

We establish the desired isomorphism as follows:
\begin{align*}
&\ \restrict{\heapcomput{\tau'}}{\validof{\tau'}}\\
\text{Equation~\eqref{Equation:FreeFormTau}}=&\ \restrict{\heapcomput{\tau}}{\validof{\tau}\setminus\invalidof{\anadr}}\\
=&\ \restrict{(\restrict{\heapcomput{\tau}}{\validof{\tau}})}{\validof{\tau}\setminus\invalidof{\anadr}}\\
\text{Ind. hypothesis, Equation~\eqref{Equation:IsoInvalid}, Lemma~\ref{Lemma:HeapIsoAlgebra}\eqref{Equation:HeapIsoRestrict}}\heapiso&\
\restrict{(\restrict{\heapcomput{\sigma}}{\validof{\sigma}})}{\validof{\sigma}\setminus\invalidof{\anisoa(\anadr)}}\\
=&\ \restrict{\heapcomput{\sigma}}{\validof{\sigma}\setminus\invalidof{\anisoa(\anadr)}}\\
\text{Equation~\eqref{Equation:FreeFormSigma}}=&\ \restrict{\heapcomput{\sigma'}}{\validof{\sigma'}}\ .
\end{align*}
\bfemph{Case (Asgn) valid}\quad
Consider $\tau'=\tau.(\athread, \psel{\apvar}:=\apvarp, [\psel{\anadr}\mapsto \anadrp])$.
Since the assignment is enabled, we have $\heapcomputof{\tau}{\apvar}=\anadr\neq \segval$.
Since the computation is PRF, we have $\apvar\in\validof{\tau}$.
Pointer $\apvarp$ has value $\heapcomputof{\tau}{\apvarp}=\anadrp$, which may be $\segval$.
Assume $\apvarp\in \validof{\tau}$. 
In this case, we have
\begin{align*}
\validof{\tau'}&=\validof{\tau}\cup\set{\psel{\anadr}}
\end{align*}
and hence
\begin{align}
\restrict{\heapcomput{\tau'}}{\validof{\tau'}}
=&\ \restrict{\heapcomput{\tau}[\psel{\anadr}\mapsto \anadrp]}{\validof{\tau}\cup\set{\psel{\anadr}}} \notag\\
=&\ (\restrict{\heapcomput{\tau}}{\validof{\tau}})[\psel{\anadr}\mapsto \anadrp].\label{Equation:AsgnValidFormTau}
\end{align}
The first equality is by definition of $\tau'$ and $\validof{\tau'}$.
The second equality uses the fact that $\apvarp$ is valid and points to $\anadrp$.
This means we preserve $\dsel{\anadrp}$ in $\restrict{\heapcomput{\tau}}{\validof{\tau}}$ (provided $\anadrp\neq\segval$) and only have to adapt the mapping of $\psel{\anadr}$. The situation may be contrasted with the case of (Malloc2) where we had to maintain $\dsel{\anadr}$.

To mimic the command, observe that $\heapcomputof{\tau}{\apvar}=\anadr\neq \segval$, $\apvar\in\validof{\tau}$, and $\restrict{\heapcomput{\tau}}{\validof{\tau}}\heapiso\restrict{\heapcomput{\sigma}}{\validof{\sigma}}$.
Together, this yields 
$\heapcomputof{\sigma}{\apvar}=\anisoa(\anadr)\neq \segval$ and allows us to dereference the address. 
Again due to isomorphism, $\apvarp$ has to be valid in $\sigma$ and $\heapcomputof{\sigma}{\apvarp}=\anisoa(\anadrp)$.
We thus get
$\sigma':=\sigma.(\athread, \psel{\apvar}:=\apvarp, [\psel{\anisoa(\anadr)}\mapsto 
\anisoa(\anadrp)])$.

By definition, 
\begin{align*}
\validof{\sigma'}&=\validof{\sigma}\cup\set{\psel{\anisoa(\anadr)}}
\end{align*}
and with the same argument as for $\tau'$
\begin{align}
\restrict{\heapcomput{\sigma'}}{\validof{\sigma'}}
=&\ \restrict{\heapcomput{\sigma}[\psel{\anisoa(\anadr)}\mapsto \anisoa(\anadrp)]}{\validof{\sigma}\cup\set{\psel{\anisoa(\anadr)}}} \notag\\
=&\ (\restrict{\heapcomput{\sigma}}{\validof{\sigma}})[\psel{\anisoa(\anadr)}\mapsto \anisoa(\anadrp)].\label{Equation:AsgnValidFormSigma}
\end{align}
The desired isomorphism $\restrict{\heapcomput{\tau'}}{\validof{\tau'}}\heapiso\restrict{\heapcomput{\sigma'}}{\validof{\sigma'}}$ now follows with Lemma~\ref{Lemma:HeapIsoAlgebra}\eqref{Equation:HeapIsoModifyPointer} in combination with the above Equations~\eqref{Equation:AsgnValidFormTau} and~\eqref{Equation:AsgnValidFormSigma}.
\\[0.2cm]
\bfemph{Case (Asgn) invalid}\quad
Let $\tau'=\tau.(\athread, \psel{\apvar}:=\apvarp, [\psel{\anadr}\mapsto \anadrp])$.
As in the previous case, by enabledness $\heapcomputof{\tau}{\apvar}=\anadr\neq \segval$ and by PRF $\apvar\in\validof{\tau}$.
Again $\heapcomputof{\tau}{\apvarp}=\anadrp$ may be $\segval$.
We now assume $\apvarp\notin \validof{\tau}$. 
This gives 
\begin{align*}
\validof{\tau'}&=\validof{\tau}\setminus\set{\psel{\anadr}}.
\end{align*}
As a result, we have
\begin{align}
\restrict{\heapcomput{\tau'}}{\validof{\tau'}}
=&\ \restrict{\heapcomput{\tau}[\psel{\anadr}\mapsto \anadrp]}{\validof{\tau}\setminus\set{\psel{\anadr}}} \notag\\
=&\ \restrict{\heapcomput{\tau}}{\validof{\tau}\setminus\set{\psel{\anadr}}}.\label{Equation:AsgnInvalidFormTau}
\end{align}
The first equality is by definition.
For the second equality, note that $\psel{\anadr}$ may already be defined in $\heapcomput{\tau}$.
Therefore, we have to remove the pointer explicitly also from this heap. 

To mimic the command, we again deduce 
$\heapcomputof{\sigma}{\apvar}=\anisoa(\anadr)\neq \segval$.
Since $\apvarp$ is not valid in $\tau$, it cannot be valid in $\sigma$ due to the isomorphism between $\restrict{\heapcomput{\tau}}{\validof{\tau}}$ and $\restrict{\heapcomput{\sigma}}{\validof{\sigma}}$. 
Let the value be $\heapcomputof{\sigma}{\apvarp}=\anadrpp$.
We thus obtain the computation 
$\sigma':=\sigma.(\athread, \psel{\apvar}:=\apvarp, [\psel{\anisoa(\anadr)}\mapsto 
\anadrpp])$.

Like  in the case of $\tau$, we have
\begin{align*}
\validof{\sigma'}&=\validof{\sigma}\setminus\set{\psel{\anisoa(\anadr)}}
\end{align*}
and hence
\begin{align}
\restrict{\heapcomput{\sigma'}}{\validof{\sigma'}}
=&\ \restrict{\heapcomput{\sigma}[\psel{\anisoa(\anadr)}\mapsto \anadrpp]}{\validof{\sigma}\setminus\set{\psel{\anisoa(\anadr)}}} \notag\\
=&\ \restrict{\heapcomput{\sigma}}{\validof{\sigma}\setminus\set{\psel{\anisoa(\anadr)}}}.\label{Equation:AsgnInvalidFormSigma}
\end{align}

We derive the desired isomorphism with Lemma~\ref{Lemma:HeapIsoAlgebra}\eqref{Equation:HeapIsoRestrict} in combination with Equations~\eqref{Equation:AsgnInvalidFormTau} and~\eqref{Equation:AsgnInvalidFormSigma}:
\begin{align*}
&\ \restrict{\heapcomput{\tau'}}{\validof{\tau'}}\\
\text{Equation~\eqref{Equation:AsgnInvalidFormTau}}=&\ \restrict{\heapcomput{\tau}}{\validof{\tau}\setminus\set{\psel{\anadr}}}\\
=&\
\restrict{(\restrict{\heapcomput{\tau}}{\validof{\tau}})}{\validof{\tau}\setminus\set{\psel{\anadr}}}\\
\text{Ind. hypothesis, Lemma~\ref{Lemma:HeapIsoAlgebra}\eqref{Equation:HeapIsoRestrict}} \heapiso&\
\restrict{(\restrict{\heapcomput{\sigma}}{\validof{\sigma}})}{\validof{\sigma}\setminus\set{\psel{\anisoa(\anadr)}}}\\
=&\ \restrict{\heapcomput{\sigma}}{\validof{\sigma}\setminus\set{\psel{\anisoa(\anadr)}}}\\
\text{Equation~\eqref{Equation:AsgnInvalidFormSigma}}=&\ \restrict{\heapcomput{\sigma'}}{\validof{\sigma'}}.
\end{align*}
\qed
\end{proof}
\begin{proof}[of Theorem~\ref{Theorem:CheckPRF}]
The implication from left to right is due to the fact that $\mmsem{\aprog}\supseteq \gcsem{\aprog}$.
For the reverse implication, we assume $\mmsem{\aprog}$ has a pointer race and from this construction a pointer race in $\gcsem{\aprog}$.
If the memory-managed semantics has a pointer race, then it has a shortest one.
Let it be $\tau.\anact\in\mmsem{\aprog}$ with $\anact$ an access to an invalid pointer.
We remove $\anact$ and obtain $\tau\in\mmsem{\aprog}$.
Membership holds as the memory-managed semantics is prefix-closed.
As $\tau$ is shorter than $\tau.\anact$, it is PRF.
This allows us to apply Proposition~\ref{Proposition:PRFImpliesGC}: There is a computation $\sigma\in\gcsem{\aprog}$ in the garbage-collected semantics with $\sigma\heapequiv \tau$.
Note that $\sigma$ is again PRF by minimality of $\tau.\anact$ and $\sigma\in \mmsem{\aprog}$.

By definition of heap equivalence, the commands in $\tau$ and $\sigma$ coincide.
This means they lead to the same control location. 
So the two semantics are, up to enabledness, ready to execute the same next command.
We moreover have $\restrict{\heapcomput{\tau}}{\validof{\tau}}\heapiso\restrict{\heapcomput{\sigma}}{\validof{\sigma}}$.
We now show how to mimic $\anact$ in the garbage-collected semantics in a way that also raises a pointer race.\\[0.2cm]
\bfemph{Case Free}\quad
Let $\anact=(\athread, \freeof{\apvar}, \emptyset)$.
Since $\tau.\anact$ is a pointer race, we have $\apvar\notin \validof{\tau}$.
Since $\anisoe$ defines a bijection between the pointers in 
$\restrict{\heapcomput{\tau}}{\validof{\tau}}$ and in 
$\restrict{\heapcomput{\sigma}}{\validof{\sigma}}$, we conclude that $\apvar\notin \validof{\sigma}$.
This means computation
\begin{align*}
\sigma.(\athread, \freeof{\apvar}, \emptyset)\in\gcsem{\aprog}\ .
\end{align*}
is also a pointer race --- as required.
\\[0.2cm]
\bfemph{Case Assignment}\quad
Let $\anact=(\athread, \acom, \anup)$ where $\acom$ is $\apvarp:=\psel{\apvar}$ with $\apvar\notin \validof{\tau}$.
As before, we conclude that $\apvar\notin\validof{\sigma}$.
With Lemma~\ref{Lemma:GCPRF}(i), we obtain $\heapcomputof{\sigma}{\apvar}=\anadr\neq \segval$.
This means we are able to dereference the address and 
can execute command $\acom$ in the garbage-collected semantics:
\begin{align*}
\sigma.(\athread, \acom, \anup')\in \gcsem{\aprog}\ .
\end{align*}
The update may differ due to the use of invalid pointers.
However, the computation will again use $\psel{\apvar}$ and, since $\apvar\notin\validof{\sigma}$, will again be racy.\\[0.2cm]
\bfemph{Case Assertion}\quad
Let $\anact=(\athread, \assert\ \acond, \emptyset)$ where $\acond$ contains $\apvar$ with $\apvar\notin\validof{\tau}$. 
As before, we derive $\apvar\notin \validof{\sigma}$. 
We are not guaranteed that the valuations in $\heapcomput{\sigma}$ and in $\heapcomput{\tau}$ coincide.
The definition of programs, however, ensures assert commands have complements.
This means if $\assert\ \acond$ is not enabled after $\sigma$, then $\assert\ \neg \acond$ will be ready for execution. 
We thus have
\begin{align*}
\sigma.(\athread, \assert\ (\neg)\acond, \emptyset)\in \gcsem{\aprog}\ .
\end{align*}
Since condition $\acond$ coincides, it will again make use of $\apvar$ with $\apvar\notin\validof{\sigma}$.
This means the computation is again racy.\qed
\end{proof}
\begin{proof}[of Lemma~\ref{Lemma:GCValid}]
If $\anadr\in\freedof{\sigma}$, then the pointers to it cannot be valid by Lemma~\ref{Lemma:Freed}.
Assume that $\anadr\notin\freedof{\sigma}$.
In the presence of garbage collection, the set $\freedof{\sigma}$ monotonically increases as $\sigma$ gets longer. 
This means $\anadr$ has not been freed throughout the computation.
We now show that there cannot be an invalid pointer to $\anadr$. 
There are two ways of creating an invalid pointer to  $\anadr$: 
Either by assigning it an invalid pointer $\apexp$ or by freeing the address. 
In particular would the first (in a shortest prefix) invalid $\apexp$ have to stem from a free on $\anadr$.
Since the address has never seen a free, there is no invalid pointer to it.\qed
\end{proof}
Proposition~\ref{Proposition:Instrumentation} follows from the following characterization of PRF under garbage collection.

\begin{lemma}\label{Lemma:NormalFormRaces}
$\gcsem{\aprog}$ is PRF if and only if there is no $\sigma_1.\anact_1.\sigma_2.\anact_2\in \gcsem{\aprog}$ so that $\comof{\anact_1}$ is $\freeof{\apvar}$ with $\heapcomputof{\sigma_1}{\apvar}=\anadr\neq \segval$ and 
$\comof{\anact_2}$ involves $\psel{\apvarp}, \dsel{\apvarp}$, $\freeof{\apvarp}$, or is an assertion with $\apvarp$, and $\heapcomputof{\sigma_1.\anact_1.\sigma_2}{\apvarp}=\anadr$.
\end{lemma}
\begin{proof}[of Lemma~\ref{Lemma:NormalFormRaces}]
For the only-if, we show the contrapositive. 
Note that $\anadr\in\freedof{\sigma_1.\anact_1.\sigma_2}$ by monotonicity of $\freedof{\sigma}$ under garbage collection. 
With $\heapcomputof{\sigma_1.\anact_1.\sigma_2}{\apvarp}=\anadr$ and Lemma~\ref{Lemma:GCValid}, we conclude $\apvarp\notin\validof{\sigma_1.\anact_1.\sigma_2}$.
Moreover, pointer variable $\apvarp$ is used in a way that raises a pointer race.

For the if-direction, we again reason by contraposition. Consider a shortest pointer race $\sigma.\anact\in\gcsem{\aprog}$.
Then there is an invalid pointer $\apvarp\notin\validof{\sigma}$ that is used in $\anact$ in a way that raises a pointer race.
Since computation $\sigma$ is shorter than $\sigma.\anact$, it is PRF. By Lemma~\ref{Lemma:GCPRF}(i), we have $\heapcomputof{\sigma}{\apvarp}=\anadr\neq \segval$.
By Lemma~\ref{Lemma:GCValid}, we conclude that $\anadr$ has been freed somewhere in $\sigma$.\qed
\end{proof}


\section{Proofs in Section~\ref{Section:SPRF}}\label{Appendix:CheckSPRF}

\begin{proof}[of Lemma~\ref{Lemma:OwnImpliesLocal} (Sketch)]
We show the contrapositive and assume that $\heapcomputof{\tau}{\apvar}\in\ownedof{\athread}{\tau}$ but (i) $\apvar\in\shared$ or (ii) $\apvar\in\localof{\athread'}$ with $\athread'\neq\athread$.
From this we derive $\apvar\notin\validof{\tau}$. 
Consider Case~(i).  
If the owning thread had passed the address via a valid pointer to $\apvar$ (potentially transitively via other public pointers, but we refrain from doing this case distinction), then $\athread$ would have lost ownership of the cell. 
As a consequence, either (i.i) $\athread$ never passed the address to $\apvar$ or (i.ii) it did so via an invalid pointer or an invalid next selector. 
In the former Case~(i.i), $\apvar$ is a dangling pointer to a cell that has been re-allocated, which in particular means $\apvar\notin\validof{\tau}$.
In the latter Case~(i.ii), the invalid right-hand side $\apt$ of the assignment $\apvar:=\apt$ will have rendered $\apvar$ invalid.
Consider now Case~(ii).
We note that threads do not assign their local pointers to the local pointers of other threads.
Therefore, the only way $\athread'$ could point to an owned cell of $\athread$ is by (ii.i) being a dangling pointer or (ii.ii) having received the reference from a shared pointer.
In the former Case~(ii.i), we immediately have $\apvar\notin\validof{\tau}$ like in Case~(i.i). 
In the latter Case~(ii.ii), the argumentation from Case~(i) shows that the shared pointer has to be invalid.
As a consequence, also $\apvar$ that receives the content of the shared pointer becomes invalid.\qed
\end{proof}

\begin{proof}[of Lemma~\ref{Lemma:OwnershipViolationImpliesSPR}]
If $\tau.(\athread, \acom, \anact)\in\mmsem{\aprog}$ violates ownership, then $\acom$ is an assignment as in Definition~\ref{Definition:OwnershipViolation}.
Let the variable be $\apvarp$ with $\heapcomputof{\tau}{\apvarp}\in\ownedof{\athread'}{\tau}$ and ($\athread'\neq \athread$ or $\apvarp\in\shared$). 
By Lemma~\ref{Lemma:OwnImpliesLocal}, the pointer cannot be valid, $\apvarp\notin\validof{\tau}$. 
Combined, we obtain the definition of SPR, Definition~\ref{Definition:SPRF}.(i).\qed
\end{proof}

\begin{proof}[of Theorem~\ref{Theorem:SPRFGuarantee}]
The ownership-respecting semantics $\resmmsem{\aprog}$ is a subset of the memory-managed semantics $\mmsem{\aprog}$, so the inclusion from right to left holds without precondition.
For the reverse inclusion, assume $\mmsem{\aprog}$ is not included in $\resmmsem{\aprog}$.
Then there is $\tau\in\mmsem{\aprog}$ that violates ownership.
By Lemma~\ref{Lemma:OwnershipViolationImpliesSPR}, $\tau$ is an SPR.
This contradicts $\mmsem{\aprog}$ SPRF.
\qed
\end{proof}

\begin{proof}[of Lemma~\ref{Lemma:FreshOwn}]
We proceed by induction on the length of the computation.
In the base case, we have single actions that set data variables arbitrarily.
We mimic them identically.
In the induction step, consider the SPRF computation $\tau.\anact\in\resmmsem{\aprog}$ and assume we are given $O\subseteq \ownpof{\tau.\anact}$.
We invoke the induction hypothesis depending on $\anact$.
Since we will always execute the same command, Requirement~(1) will trivially hold and we rephrain from commenting on it.
\\[0.2cm]
\bfemph{Case (Malloc2)}\quad
Consider $\anact = (\athread, \apvar:=\malloc, [\apvar\mapsto\anadr])$, which means we re-allocate an address that has been freed.  
With this assignment, address $\anadr$ is owned by thread $\athread$ and $\apvar$ is an owning pointer, $\anadr\in \ownedof{\athread}{\tau.\anact}$ and $\ownpof{\tau.\anact}=\ownpof{\tau}\cup\set{\apvar}$.
Since we only turn $\apvar$ into an owning pointer, we can safely use $O':=O\setminus\set{\apvar}\subseteq\ownpof{\tau}$ to invoke the induction hypothesis for $\tau$.
The hypothesis returns a computation $\tau'$.
Since the sequences of commands coindice for $\tau$ and for $\tau'$, the threads reach the same control locations and hence also in $\tau'$ thread $\athread$ is ready to execute a malloc.
The result of the allocation will depend on whether or not $\apvar$ belongs to the given set $O$:
\begin{description}
\item[$\apvar\notin O$] 
We again allocate the address $\anadr$ with (Malloc2). 
The transition is enabled after $\tau'$ by $\anadr\in\freedof{\tau}\subseteq\freedof{\tau'}$.
It remains to check (2) to (6).
Concerning Requirement~(2), the only pointer outside $O$ that we change is $\apvar$, and we set it consistently to $\anadr$ in both $\tau.\anact$ and in $\tau'.\anact$.
The isomorphism in Requirement~(3) is also fine by
\begin{align*}
\project{\heapcomput{\tau.\anact}}{\validof{\tau.\anact}}
&=\project{(\heapcomput{\tau}[\apvar\mapsto\anadr])}{\validof{\tau}\cup\set{\apvar}}\\
&=(\project{\heapcomput{\tau}}{\validof{\tau}})[\apvar\mapsto\anadr,\dsel{\anadr}\mapsto\advalue]\\
&\heapiso(\project{\heapcomput{\tau'}}{\validof{\tau'}})[\apvar\mapsto\anadr,\dsel{\anadr}\mapsto\advalue]\\
&=\project{(\heapcomput{\tau'}[\apvar\mapsto\anadr])}{\validof{\tau'}\cup\set{\apvar}}\\
&=\project{\heapcomput{\tau'.\anact}}{\validof{\tau'.\anact}}
\end{align*}
where we preserve $\dsel{\anadr}$ as in Proposition~\ref{Proposition:PRFImpliesGC}.
For Requirement~(4), we consistently remove $\anadr$ from the set of freed addresses in $\tau.\anact$ and in $\tau'.\anact$.
For the owning pointers in Requirement~(5) we have
\begin{align*}
\ownpof{\tau'.\anact}
=&\ \ownpof{\tau'}\cup\set{\apvar}\\
=&\ (\ownpof{\tau}\setminus O')\ \cup\ \fune(O')\cup\set{\apvar}\\
=&\ ((\ownpof{\tau}\cup\set{\apvar})\setminus O)\ \cup\ \fune(O)\\
=&\ (\ownpof{\tau.\anact}\setminus O)\ \cup\ \fune(O).
\end{align*}
The first equation is by definition, the second is the hypothesis for $\tau$ and $O'$, the third equation holds by $O'=O$ as $\apvar$ is assumed not to belong to $O$, and the last equation is again by definition. 
Concerning freshness of the owning pointers, Requirement~(6),
we note $O'=O$. 
This gives 
\begin{align*}
	 	\heapcomputof{\tau'.\anact}{\fune(O)}
	=	\heapcomput{\tau'}[\apvar\mapsto\anadr](\fune(O'))
	=	\heapcomputof{\tau'}{\fune(O')}.
\end{align*}
The latter equality is because $\apvar\notin O'$ and hence $\apvar\notin\fune(O')$.
We have $\apvar\notin\fune(O')$ as only $\fune(\apvar)=\apvar$. 
Moreover, $\adrof{\heapcomput{\tau.\anact}}=\adrof{\heapcomput{\tau}}\cup\set{\anadr}=\adrof{\heapcomput{\tau}}$.
The latter equality holds because $\dsel{\anadr}$ is defined in $\heapcomput{\tau}$ and we re-allocate the address.
The hypothesis, $\adrof{\heapcomput{\tau}}\cap\heapcomputof{\tau'}{\fune(O')}=\emptyset$, combined with the previous argumentation yields $\adrof{\heapcomput{\tau.\anact}}\cap\heapcomputof{\tau'.\anact}{\fune(O)}=\emptyset$.
\item[$\apvar\in O$] 
We allocate a fresh address using (Malloc1) and $\anact'=(\athread, \apvar:=\malloc, \anup)$ with $\anup=[\apvar\mapsto \anadrp, \dsel{\anadrp}\mapsto\advalue,\setcond{\psel{\anadrp}\mapsto\segval}{\text{for every selector $\mathtt{next}$}}]$. 
Requirement~(2) holds as we only change a pointer in $O$.
Requirement~(3) is as in the proof of Proposition~\ref{Proposition:PRFImpliesGC}.
Requirement~(4) holds by
\begin{align*}
\freedof{\tau.\anact}\ =\ \freedof{\tau}\setminus\set{\anadr}\ \subseteq\ \freedof{\tau'}\ =\ \freedof{\tau'.\anact'},
\end{align*}
where the inclusion is due to the hypothesis.

Concerning the address mapping, we set it to $\funa\disunion\set{\anadr\mapsto\anadrp}$.
It remains to check that this is a function, which means $\anadr\notin\domof{\funa}=\adrof{O}$. 
We have $\anadr\in\freedof{\tau}$, hence there is no valid pointer to this address and no valid next selector defined at this address by Lemma~\ref{Lemma:Freed}. 
Since all pointers in $O$ are valid, the claim follows. 
Actually, Lemma~\ref{Lemma:Freed} assumes the computation to be PRF, but an inspection of the proof shows that it continues to hold for SPRF computations.

For Requirement~(5), we have
\begin{align*}
\ownpof{\tau'.\anact'}&=\ownpof{\tau'}\cup\set{\apvar}\\
&=(\ownpof{\tau}\setminus O')\cup \fune(O')\cup\set{\apvar}\\
&=(\ownpof{\tau}\setminus O')\cup \fune(O)\\
&=(\ownpof{\tau.\anact}\setminus O)\cup \fune(O).
\end{align*}
The first equation is by definition, the second invokes the hypothesis for $\tau$ and $O'$.
The third equation uses the fact that $\fune(\apvar)=\apvar$.
In the last equation, we add $\apvar$ to $\ownpof{\tau}$ and $O'$. 
For Requirement~(6), $
\adrof{\heapcomput{\tau.\anact}}=\adrof{\heapcomput{\tau}}\cup\set{\anadr}=\adrof{\heapcomput{\tau}}$. 
The latter equation holds because $\anadr$ is re-allocated and thus $\dsel{\anadr}$ is defined in $\heapcomput{\tau}$. 
Moreover, we have 
\begin{align*}
	 	\heapcomputof{\tau'.\anact'}{\fune(O)}
	=	\heapcomput{\tau'}[\apvar\mapsto\anadrp](\fune(O))
	=	\heapcomput{\tau'}(\fune(O'))\cup\set{\anadrp}.
\end{align*}
An application of the induction hypothesis gives $\adrof{\heapcomput{\tau}}\cap\heapcomputof{\tau'}{\fune(O')}=\emptyset$.
Since $\anadrp$ is fresh, the required disjointness $\adrof{\heapcomput{\tau.\anact}}\cap\heapcomputof{\tau'.\anact}{\fune(O)}=\emptyset$ holds.
\end{description}
\bfemph{Case (Malloc1)}\quad
Consider the case that $\apvar$ allocates a fresh address $\anadr$ using (Malloc1).
Like in the previous case, we invoke the induction hypothesis for $\tau$ with $O':=O\setminus\set{\apvar}$ and obtain $\tau'$. 
We can safely assume that $\anadr$ has not been allocated in $\tau'$. 
Indeed, in cases where $\tau'$  deviates from $\tau$,  it can allocate fresh addresses different from $\anadr$.
We mimic the allocation of $\apvar$ in $\tau'$ with another invocation to (Malloc1).
Depending on whether $\apvar\notin O$ or $\apvar\in O$, the mimicking allocation also selects $\anadr$ or it selects a fresh $\anadrp$, respectively.
If $\apvar\in O$, the address mapping is extended by $\anadr\mapsto\anadrp$. 
The Requirements~(2) to (4) are immediate. 
Requirement~(5) is checked like in the previous case.
Requirement~(6) is by the induction hypothesis and the fact that $\anadrp$ is chosen fresh. \\[0.2cm]
\bfemph{Case (Free)}\quad
Consider $\anact = (\athread, \freeof{\apvar}, \emptyset)$.
We note that $O\subseteq \ownpof{\tau}$ since $\ownpof{\tau.\anact}=\ownpof{\tau}\setminus\invalidof{\heapcomputof{\tau}{\apvar}}$.
This allows us to invoke the hypothesis and obtain the computation $\tau'$.
After $\tau'$ we are ready to execute the same action $\anact$ as in the computation $\tau$.
To establish $\tau'.\anact\in \resmmsem{\aprog}$, we have to show that the computation respects ownership.
Towards a contradiction, assume $\heapcomputof{\tau'}{\apvar}\in \ownedof{\tau'}{\athread'}$ with $\athread'\neq \athread$ or $\apvar$ shared.
With Lemma~\ref{Lemma:OwnImpliesLocal}, we have that $\apvar$ cannot be valid in $\tau'$.
Since the valid pointers in $\tau$ and $\tau'$ coincide by Requirement~(3) in the induction hypothesis, we have that $\apvar$ was not valid in $\tau$. 
But this in turn means that action $\anact$ frees an invalid pointer in $\tau$, which raises an SPR.
A contradiction to the assumption that $\tau.\anact$ is SPRF.

It remains to check the guarantees (2) to (6) required by the induction.
For Requirement~(2) note that $\freeof{\apvar}$ does not execute any updates.
Hence $\heapcomputof{\tau.\anact}=\heapcomputof{\tau}$ and similarly for $\tau'$.
Hence it remains to apply the induction hypothesis as follows:
\begin{align*}
	  	\restrict{\heapcomput{\tau.act}}{\pexp\setminus O}
	&=	\restrict{\heapcomput{\tau}}{\pexp\setminus O} \\
	&=	\restrict{\heapcomput{\tau'}}{\pexp\setminus\fune(O)} \\
	&=	\restrict{\heapcomput{\tau'.act}}{\pexp\setminus\fune(O)}.		
\end{align*}
Requirement~(3) is like in Proposition~\ref{Proposition:PRFImpliesGC}.
For Requirement~(4), note that the induction hypothesis guarantees $\heapcomputof{\tau}{\apvar}=\heapcomputof{\tau'}{\apvar}$ because $\apvar\notin O$. 
Moreover, the hypothesis gives $\freedof{\tau}\subseteq\freedof{\tau'}$.
Together, this yields
\begin{align*}
\freedof{\tau.\anact}=\freedof{\tau}\cup\set{\heapcomputof{\tau}{\apvar}}\subseteq \freedof{\tau'}\cup\set{\heapcomputof{\tau'}{\apvar}}=\freedof{\tau'.\anact}.
\end{align*}
Concerning Property~(5), we note that
\begin{align*}
\ownpof{\tau'.\anact}
&=\ownpof{\tau'}\setminus\invalidof{\heapcomputof{\tau'}{\apvar}}\\
&=((\ownpof{\tau}\setminus O)\cup \fune(O))\setminus\invalidof{\heapcomputof{\tau'}{\apvar}}\\
&=((\ownpof{\tau}\setminus O)\cup \fune(O))\setminus\invalidof{\heapcomputof{\tau}{\apvar}}\\
&=((\ownpof{\tau}\setminus O)\setminus\invalidof{\heapcomputof{\tau}{\apvar}})\cup \fune(O)\\
&=((\ownpof{\tau}\setminus\invalidof{\heapcomputof{\tau}{\apvar}})\setminus O)\cup \fune(O)\\
&=(\ownpof{\tau.\anact}\setminus O)\cup \fune(O).
\end{align*}
The first equality is by the definition of owning pointer, more precisely, the requirement that they have to be valid. 
The second equation is the induction hypothesis for $\tau$ and $O$.

For the third equality, we show $\invalidof{\heapcomputof{\tau}{\apvar}}=\invalidof{\heapcomputof{\tau'}{\apvar}}$. 
We argue for $\subseteq$.
We have $\apvar\notin O$.
Since $O$ is chosen coherent, for every other pointer $\apexp=\apvarp$ or $\apexp=\psel{\anadr}$ to $\heapcomputof{\tau}{\apvar}$, we have $\apexp\notin O$.  
With the induction hypothesis, Requirement~(2), all these pointers $\apexp$ are mapped identically, 
\begin{align*}
\heapcomputof{\tau'}{\apexp}=\heapcomputof{\tau}{\apexp}=\heapcomputof{\tau}{\apvar}=\heapcomputof{\tau'}{\apvar}.
\end{align*}
So $\apexp\in\invalidof{\heapcomputof{\tau'}{\apvar}}$.  
For the next pointers $\psel{\heapcomputof{\tau}{\apvar}}$ that are turned invalid by the free in $\tau$, membership in $\invalidof{\heapcomputof{\tau'}{\apvar}}$ is by $\heapcomputof{\tau}{\apvar}=\heapcomputof{\tau'}{\apvar}$. 
For the reverse inclusion $\supseteq$, we argue towards a contradiction and assume that it does not hold.
Certainly, every next pointer $\psel{\heapcomputof{\tau'}{\apvar}}$ is also in $\invalidof{\heapcomputof{\tau}{\apvar}}$ by $\heapcomputof{\tau'}{\apvar}=\heapcomputof{\tau}{\apvar}$. 
So there is a pointer $\apexp$ to $\heapcomputof{\tau'}{\apvar}$ that is not in $\invalidof{\heapcomputof{\tau}{\apvar}}$.
Since by the hypothesis, Requirement~(2), the heaps coincide except for $O$ and $\fune(O)$, this other pointer has to belong to $\fune(O)$.
But then $\apexp$ cannot point to $\heapcomputof{\tau'}{\apvar}=\heapcomputof{\tau}{\apvar}$, because the address also belongs to $\adrof{\heapcomput{\tau}}$. 
This would violate Requirement~(6) in the induction hypothesis.

For the fourth equation, we argue that $\fune(O)\cap 
\invalidof{\heapcomputof{\tau}{\apvar}}=\emptyset$. 
Assume there was a pointer $\apexp\in\fune(O)$ with $\apexp\in \invalidof{\heapcomputof{\tau}{\apvar}}$.
Then $\apexp$ points to $\heapcomputof{\tau}{\apvar}$ or it has the shape $\psel{\heapcomputof{\tau}{\apvar}}$.
In both cases, $\apexp\notin O$ since $\apvar\notin O$ and $O$ is chosen coherent. 
In the former case, we get $\heapcomputof{\tau'}{\apexp}=\heapcomputof{\tau}{\apexp}$ by Requirement~(2) in the hypothesis.  
This contradicts the disjointness in Requirement~(6) in the hypothesis. 
In the latter case, we have $\heapcomputof{\tau'}{\apvar}=\heapcomputof{\tau}{\apvar}$, and hence $\heapcomputof{\tau}{\apvar}\in \heapcomputof{\tau'}{\fune(O)}$ also violates the disjointness in Requirement~(6) in the induction hypothesis.

The fifth equation is set theory. 
The last equation is the definition of the set of owning pointers.

Requirement~(6) follows from the induction hypothesis and the fact that the free does not change the heap.
\\[0.2cm]
\bfemph{Case (Asgn) owned}\quad
Consider an assignment $\anact=(\athread, \psel{\apvar}:=\apvarp, [\psel{\anadr}\mapsto \anadrp])$.
We consider the case that $\psel{\anadr}\in O$. 
To invoke the hypothesis, we choose the largest coherent subset $O'\subseteq O$ with $O'\subseteq \ownpof{\tau}$.
Let the resulting computation be $\tau'$.

As a first step, we argue that $\apvarp\in O'$.
We have $\ownpof{\tau.\anact}=\ownpof{\tau}\cup\set{\psel{\anadr}}$. 
Hence, $O'=O\setminus\set{\psel{\anadr}}$ if $\heapcomputof{\tau}{\psel{\anadr}}$ is not the target or source address of a pointer in $O$, and $O'=O$ otherwise.
We have $\apvarp\in O$ since $\heapcomputof{\tau.\anact}{\apvarp}=\heapcomputof{\tau.\anact}{\psel{\anadr}}$, $\psel{\anadr}\in O$, and $O$ is coherent.
Moreover, $\apvarp\in\ownpof{\tau}$, for otherwise $\psel{\anadr}$ would not have become an owning pointer to $\anadrp$ after the assignment.
Hence, $\apvarp\in O'$.

To mimic the command, note that by the induction hypothesis we are in the same control state and thus ready to execute the assigment. 
To make sure the assignment is enabled, we check that $\heapcomputof{\tau'}{\apvar}\neq \segval$.
This holds by the hypothesis, Requirements~(2) and (3).
Indeed, if $\apvar\notin O'$, then the address is mapped identically, and we get $\heapcomputof{\tau'}{\apvar}=\anadr=\heapcomputof{\tau}{\apvar}$.
If $\apvar\in O'$, then the address is mapped by $\funa$.
Since $\funa$ is an address mapping, it only maps $\segval$ to $\segval$, and hence $\anadr\neq\segval$ to $\heapcomputof{\tau'}{\apvar}=\funa(\anadr)\neq \segval$. 
With this argument, we obtain $\tau'.\anact'$ with $\anact':=(\athread, \psel{\apvar}:=\apvarp, [\heapcomputof{\tau'}{\apvar}\mapsto\funa(\anadrp)])$. 

To establish $\tau'.\anact'\in\resmmsem{\aprog}$, we have to show that the assignment in $\anact'$ respects ownership.
Towards a contradiction, assume $\heapcomputof{\tau'}{\apvar}\in \ownedof{\tau'}{\athread'}$ with $\athread'\neq \athread$ or $\apvar$ shared. 
By Lemma~\ref{Lemma:OwnImpliesLocal}, $\apvar$ cannot be valid in $\tau'$.
Since the valid pointers in $\tau'$ and in $\tau$ coincide, Requirement~(3) in the induction hypothesis, $\apvar$ cannot be valid in $\tau$.
By Definition~\ref{Definition:SPRF}, the assignment $\psel{\apvar}:=\apvarp$ would raise an SPR --- in contradiction to the assumption that $\tau.\anact$ is SPRF.

We now show that the new computation satisfies the Requirements~(2) to (6).
For Requirement~(2), we obtain equality by the induction hypothesis.
Indeed, we may only add $\psel{\anadr}$ to $O$ if it did not belong to $O'$ before the assignment.
Requirement~(3) is as in Proposition~\ref{Proposition:PRFImpliesGC}.
The freed addresses do not change by the assignment, therefore, for Requirement~(4) there is nothing to do.
For Requirement~(5), we note that 
\begin{align*}
\ownpof{\tau'.\anact'}&=\ownpof{\tau'}\cup\set{\psel{\heapcomputof{\tau'}{\apvar}}}\\
&=(\ownpof{\tau}\setminus O')\cup \fune(O')\cup\set{\psel{\heapcomputof{\tau'}{\apvar}}}.
\end{align*}
There are two cases. 
Assume $O'=O$ or, phrased differently, 
$\psel{\heapcomputof{\tau}{\apvar}}\in O'$.
Then we have $\psel{\heapcomputof{\tau'}{\apvar}}\in \fune(O')$
by Requirement~(2) in the hypothesis. Thus,
\begin{align*}
&(\ownpof{\tau}\setminus O')\cup \fune(O')\cup\set{\psel{\heapcomputof{\tau'}{\apvar}}}\\
=&((\ownpof{\tau}\cup\set{\psel{\heapcomputof{\tau}{\apvar}}})\setminus (O'\cup\set{\psel{\heapcomputof{\tau}{\apvar}}}))\ \cup\ \fune(O')\\
=&(\ownpof{\tau.\anact}\setminus O)\cup \fune(O).
\end{align*}
The first equation is set theory and the fact that $\psel{\heapcomputof{\tau'}{\apvar}}\in \fune(O')$.
The second equation is by definition of $\ownpof{\tau.\anact}$, the fact that $\psel{\heapcomputof{\tau}{\apvar}}\in O'$, and by $O'=O$. 

Assume $O'=O\setminus\set{\psel{\heapcomputof{\tau}{\apvar}}}$.
For the induction step, we update $\funa$ by mapping $\heapcomputof{\tau}{\apvar}$ identically.  
We have $\heapcomputof{\tau'}{\apvar}=\heapcomputof{\tau}{\apvar}$, for otherwise $\apvar\in O'$ and by coherence $\psel{\heapcomputof{\tau}{\apvar}}\in O'$.
From this, we derive
\begin{align*}
&(\ownpof{\tau}\setminus O')\cup \fune(O')\cup\set{\psel{\heapcomputof{\tau'}{\apvar}}}\\
=&((\ownpof{\tau}\cup\set{\psel{\heapcomputof{\tau}{\apvar}}})\setminus (O'\cup \set{\psel{\heapcomputof{\tau}{\apvar}}}))\cup \fune(O')\cup\set{\psel{\heapcomputof{\tau'}{\apvar}}}\\
=&(\ownpof{\tau.\anact}\setminus O)\cup \fune(O')\cup\set{\psel{\heapcomputof{\tau'}{\apvar}}}\\
=&(\ownpof{\tau.\anact}\setminus O)\cup \fune(O).
\end{align*}
We have $\fune(O')\cup\set{\psel{\heapcomputof{\tau'}{\apvar}}}=
\fune(O')\cup\set{\psel{\heapcomputof{\tau}{\apvar}}}=\fune(O)$
by the choice of the address function.
Requirement~(6) is easier to check.
%
\\[0.2cm]
\bfemph{Case (Asgn) not owned}\quad
Consider an assignment $\anact=(\athread, \psel{\apvar}:=\apvarp, [\psel{\anadr}\mapsto \anadrp])$.
Assume $\psel{\apvar}\not\in O$.
Since $O$ is coherent for $\tau.act$, we also have $\apvarp\not\in O$.
We invoke the induction hypothesis for $\tau$ with $O$.
We can do so since $O$ is also coherent for $\tau$ by the definition of owned addresses.
The hypothesis yields some $\tau'\in\resmmsem{\aprog}$ and some address mapping $\funa$ satisfying Requirements~(1) to~(6).

First, we argue that $\anact$ is enabled after $\tau'$.
Therefore, we have to show that $\heapcomputof{\tau'}{\apvarp}\not=\segval$.
The argument is as in the previous case ((Asgn) owned).
We can now establish Requirements~(2) to (6).

Requirement~(2) follows from $\psel{\apvar}\not\in O$, $\apvarp\not\in O$,
$\fune({\psel{\apvar}})\not\in\fune(O)$, $\fune({\apvarp})\not\in\fune(O)$ and the induction hypothesis.
Hence, we can derive the following:
\begin{align*}
	  	\restrict{\heapcomput{\tau.act}}{\pexp\setminus O}
	&=	\restrict{\heapcomput{\tau}[\psel{\anadr}\mapsto \anadrp]}{\pexp\setminus O} \\
	&=	\restrict{\heapcomput{\tau}}{\pexp\setminus O}[\psel{\anadr}\mapsto \anadrp] \\
	&=	\restrict{\heapcomput{\tau'}}{\pexp\setminus\fune(O)}[\psel{\anadr}\mapsto \anadrp] \\
	&=	\restrict{\heapcomput{\tau'}[\psel{\anadr}\mapsto \anadrp]}{\pexp\setminus\fune(O)} \\
	&=	\restrict{\heapcomput{\tau'.act}}{\pexp\setminus\fune(O)}.
\end{align*}
Requirement~(3) is as in Proposition~\ref{Proposition:PRFImpliesGC}.
The freed addresses do not change by the assignment, therefore, Requirement~(4) holds, too.
For Requirement~(5), we observe that $\ownpof{\tau.\anact}=\ownpof{\tau}\cup X$ with $X=\ownpof{\tau.\anact}\setminus\ownpof{\tau}$, by the definition of ownership.
The analogue holds for $\ownpof{\tau'.\anact}$.
Moreover, $X\cap O=\emptyset$.
Hence, we can establish Requirement~(5) as follows:
\begin{align*}
	  	\ownpof{\tau'.\anact}
	&=	\ownpof{\tau'} \cup X \\
	&=	(\ownpof{\tau}\setminus O)\cup\fune(O)\cup X \\
	&=	((\ownpof{\tau}\cup X)\setminus O)\cup\fune(O) \\
	&=	(\ownpof{\tau.\anact}\setminus O)\cup\fune(O).
\end{align*}
Requirement~(6) follows from the fact that $\psel{\apvar}\not\in\fune(O)$.
\\[0.2cm]
\bfemph{Case (Asrt)}\quad
Consider $\anact = (\athread, \assert\ \apvar=\apvarp, \emptyset)$ with $\heapcomputof{\tau}{\apvar}=\anadr$.
By enabledness we have $\anadr=\anadrp \vee \anadr=\segval \vee \anadrp=\segval$.
We invoke the induction hypothesis on $\tau$ with $O$.
This yields some $\tau'\in\resmmsem{\aprog}$ and some address mapping $\funa$ satisfying Requirements~(1) to~(6).

We now prove that $\anact$ is enabled in $\tau'$, i.e. $\tau'.\anact\in\resmmsem{\aprog}$.
Towards a contradiction, assume that $\anact$ is not enabled in $\tau'$.
Therefore, let $\heapcomputof{\tau'}{\apvar}=\anadr'$ and $\heapcomputof{\tau'}{\apvarp}=\anadrp'$.
Since $\anact$ is not enabled, we have $\anadr'\not=\anadrp'$ with $\anadr'\not=\segval$ and $\anadrp'\not=\segval$.
Moreover, we can conclude that $\apvar,\apvarp\not\in\validof{\tau'}$ since $\restrict{\heapcomput{\tau}}{\validof{\tau}}\heapiso\restrict{\heapcomput{\tau'}}{\validof{\tau'}}$ by induction hypothesis.
From this we get $\apvar,\apvarp\not\in O$ as $O$ may only contain valid pointers by definition.
Hence, we come up with the following equalities due to Requirement~(2):
\begin{align*}
	 	a'
	=	\restrict{\heapcomput{\tau'}}{\pexp\setminus \fune(O)}(\apvar)
	=	\restrict{\heapcomput{\tau}}{\pexp\setminus O}(\apvar)
	=	a, \\
	 	b'
	=	\restrict{\heapcomput{\tau'}}{\pexp\setminus \fune(O)}(\apvarp)
	=	\restrict{\heapcomput{\tau}}{\pexp\setminus O}(\apvarp)
	=	b. \\
\end{align*}
Since $\segval\not=\anadr'\not=\anadrp'\not=\segval$ by assumption, we conclude $\segval\not=\anadr\not=\anadrp\not=\segval$.
This contradicts enabledness of $\anact$ in $\tau$.
Hence, we have proven that $\anact$ is indeed enabled in~$\tau'$.

It remains to establish Requirements~(2) to (6).
For Requirement~(5), consider the two sets $X$ and $X'$ which contain exactly those pointer expressions which $\athread$ loses ownership of by executing the assertion in $\tau$ and $\tau'$, respectively.
Formally, $X$ and $X'$ are defined as
\begin{align*}
	X  := \ownpof{\tau} \setminus \ownpof{\tau.\anact} &&
	X' := \ownpof{\tau'} \setminus \ownpof{\tau'.\anact}.
\end{align*}
By the definition of owning pointers, this is equivalent to
\begin{align*}
	X  &= \{\, \apexp ~|~ \apexp\in\validof{\tau} \wedge \heapcomput{\tau}(\apexp)=\heapcomput{\tau}(\apvar)\\&\qquad\qquad~ \wedge \heapcomput{\tau}(\apvar)\in\ownpof{\tau} \wedge \heapcomput{\tau}(\apvar)\not\in\ownpof{\tau.\anact} \,\}, \\
	X' &= \{\, \apexp ~|~ \apexp\in\validof{\tau'} \wedge \heapcomput{\tau'}(\apexp)=\heapcomput{\tau'}(\apvar)\\&\qquad\qquad~ \wedge \heapcomput{\tau'}(\apvar)\in\ownpof{\tau'} \wedge \heapcomput{\tau'}(\apvar)\not\in\ownpof{\tau'.\anact} \,\}.
\end{align*}
Now, we can easily state that both $X\cap O=\emptyset$ and $X'\cap O=\emptyset$ hold.
Hence, Requirement~(2) from the induction hypothesis gives us $\heapcomput{\tau}(\apvar)=\heapcomput{\tau'}(\apvar)$.
This ultimately implies that $X=X'$.
With this equality at hand, we can now establish Requirement~(5) as follows:
\begin{align*}
	  	\ownpof{\tau'.\anact}
	&=	\ownpof{\tau'} \setminus X' \\
	&=	\ownpof{\tau'} \setminus X \\
	&=	((\ownpof{\tau}\setminus O) \cup \fune(O)) \setminus X \\
	&=	((\ownpof{\tau}\setminus X)\setminus O) \cup \fune(O) \\
	&=	(\ownpof{\tau.\anact}\setminus O) \cup \fune(O)
\end{align*}

For the remaining Requirements note that $\heapcomput{\tau.\anact}=\heapcomput{\tau}$, $\validof{\tau.\anact}=\validof{\tau}$ and $\freedof{\tau.\anact}=\freedof{\tau}$ hold.
Furthermore, the analogues for $\tau'$ hold, too.
To prove the remaining Requirements it is now sufficient to apply the above equalities and invoke the induction hypothesis.
\qed
\end{proof}


\section{Evaluation Details} 
\label{sec:evaluation_details}

This section provides additional information about the experiments discussed in Section~\ref{Section:Evaluation}.
Figures~\ref{fig:cstack} and \ref{fig:cqueue} give the implementation of the single lock data structures \textit{coarse stack} and \textit{coarse queue}.
Moreover, Table~\ref{tab:experrors} provides experimental results for our stress tests.
Those test were conducted using the ownership-respecting semantics.
We tested whether or not our tool is able to detected purposely inserted bugs.
For each linearisation point we executed a test where we moved it to an erroneous position: once to late and once to early.
A description of the correct linearisation points can be found in Table~\ref{tab:linps}).
In addition we swapped some assignments.
In Treiber's stack (Figure~\ref{fig:treiberscode}) we moved the \texttt{free} in \texttt{pop} before the statement reading the value from the node to be freed.
In Michael\&Scott's queue (Figure~\ref{fig:mscode}) we moved the statement reading the value to be returned by \texttt{dey} after the following \texttt{CAS}.
Both swapped statements result in unsafe behavior as potentially freed cells are accessed.

\begin{figure}
	\caption{Coarse Stack}%
	\label{fig:cstack}%
	\begin{minipage}{.49\textwidth}
		\begin{lstlisting}
struct Node {
	data_type data;
	Node* next;
}

Node* ToS;

void init() {
	ToS = NULL;
}

void push(data_type val) {
	Node* node = new Node();
	node->data = val;
	atomic {
		node->next = ToS;
		ToS = node;
	}
}
		\end{lstlisting}
	\end{minipage}%
	\hfill%
	\begin{minipage}{.51\textwidth}
		\begin{lstlisting}
bool pop(data_type& dst) {
	Node* node;
	atomic {
		node = ToS;
		if (node != NULL)
			ToS = node->next;
	}
	if (node == NULL) {
		return false;
	} else {
		dst = node->data;
		delete node;
	}
}
		\end{lstlisting}
	\end{minipage}%
\end{figure}

\begin{figure}
	\caption{Coarse Queue}%
	\label{fig:cqueue}%
	\begin{minipage}{.49\textwidth}
		\begin{lstlisting}
struct Node {
	data_type data;
	Node* next;
}

Node* Head, Tail;

void init() {
	Head = new Node();
	Tail = Head;
}

void enq(data_type val) {
	Node* node = new Node();
	node->data = val;
	node->next = NULL;
	atomic {
		Tail->next = node;
		Tail = node;
	}
}
		\end{lstlisting}
	\end{minipage}%
	\hfill%
	\begin{minipage}{.51\textwidth}
		\begin{lstlisting}
bool deq(data_type& dst) {
	atomic {
		Node* node = Head;
		Node* next = Head->next;
		if (next == NULL) {
			return false;
		} else {
			// read data inside
			// the atomic block
			// to ensure that
			// no other thread
			// frees "next" in
			// between
			dst = next->data;
			Head = next;
		}
	}
	delete node;
}
		\end{lstlisting}
	\end{minipage}%
\end{figure}

\begin{figure}
	\caption{Treiber's lock-free stack with linearisation points.}%
	\label{fig:treiberscode}%
	\begin{lstlisting}
struct pointer_t { Node* ptr; int age; }
struct Node { data_type data; pointer_t next; }

pointer_t ToS;

void init() {
	ToS.ptr = NULL;
}

void push(data_type val) {
	pointer_t node;
	node.ptr = new Node();
	node.ptr->data = val;
	while (true) {
		pointer_t top = ToS;
		node.ptr->next = top;
		if (DWCAS(ToS, top, node))              // @1
			break;
	}
}

bool pop(data_type& dst) {
	while (true) {
		pointer_t top = ToS;                    // @2
		if (top.ptr == NULL) {
			return false;
		} else {
			pointer_t node = top.ptr->next;
			if (DWCAS(ToS, top, node)) {        // @3
				dst = top.ptr->data;
				delete top.ptr;
				break;
			}
		}
	}
}

bool DWCAS(pointer_t& dst, pointer_t cmp, pointer_t src) {
	atomic {
		if (dst.ptr == cmp.ptr && dst.age == cmp.age) {
			dst.ptr = src.ptr;
			dst.age = cmp.age + 1;
			return true;
		} else return false;
	}
}
	\end{lstlisting}
\end{figure}

\begin{figure}
	\caption{Michael\&Scott's lock-free queue with linearisation points.}%
	\label{fig:mscode}%
	\begin{lstlisting}
		struct pointer_t { Node* ptr; int age; }
		struct Node { data_type data; pointer_t next; }

		pointer_t Head, Tail;

		void init() {
			Head = new Node();
			Head.ptr->next = NULL;
			Tail = Head;
		}

		void enq(data_type val) {
			pointer_t node;
			node.ptr = new Node();
			node.ptr->data = val;
			node.ptr->next = NULL;
			while (true) {
				pointer_t tail = Tail;
				pointer_t next = tail.ptr->next;
				if (tail == Tail)
					if (tail.age == Tail.age)
						if (next.ptr == NULL) {
							if (DWCAS(tail.next, next, node))  // @1
								break;
						} else DWCAS(Tail, tail, next);
			}
			DWCAS(Tail, tail, node);
		}

		bool deq(data_type& dst) {
			while (true) {
				pointer_t head = Head;
				pointer_t tail = Tail;
				pointer_t next = head.ptr->next;            // @2
				if (head == Head)
					if (head.age == Head.age)
						if (head.ptr == tail.ptr) {
							if (next.ptr == NULL) {
								return false;
							}
							DWCAS(Tail, tail, next);
						} else {
							dst = next.ptr->data;
							if (DWCAS(Head, head, next)) {     // @3
								free(head);
								return true;
							}
						}
			}
		}
	\end{lstlisting}
\end{figure}

\begin{table}
	\caption{Linearisation points.}%
	\label{tab:linps}%
	\vspace{-5mm}%
	\center%
	\def\firstcolwidth{3cm}%
	\newcolumntype{Y}{>{\centering\arraybackslash}X}%
	\newcolumntype{Z}{>{\raggedright}m}%
	\newcolumntype{G}{>{\centering}m}%
	\newcolumntype{J}{>{\raggedright\let\newline\\\arraybackslash}X}%
	\begin{tabularx}{\textwidth}{Z{\firstcolwidth+.2cm}G{2.3cm}J}
		\toprule[0.1ex]
		Program & Linearisation Point & Description \\
		\midrule[0.3ex]
		\multirow{3}{\firstcolwidth}{Treiber's stack}
			& \texttt{@1}	& \texttt{CAS} in \texttt{push}                       	\\
			& \texttt{@2}	& reading global top of stack pointer in \texttt{push}	\\
			& \texttt{@3}	& \texttt{CAS} in \texttt{pop}                        	\\
		\hdashline[1pt/1pt]
		\multirow{3}{\firstcolwidth}{Michael\&Scott's queue}
			& \texttt{@1}	& \texttt{CAS} in \texttt{enq} adding new node to the tail        	\\
			& \texttt{@2}	& reading next field of head of queue in \texttt{deq}             	\\
			& \texttt{@3}	& \texttt{CAS} in \texttt{deq} moving global head of queue pointer	\\
		\bottomrule[0.1ex]
	\end{tabularx}
\end{table}

\begin{table}
	\caption{Experimental results for erroneous programs.}%
	\label{tab:experrors}%
	\vspace{-5mm}%
	\center%
	\def\firstcolwidth{3cm}%
	\newcolumntype{Y}{>{\centering\arraybackslash}X}%
	\newcolumntype{Z}{>{\raggedright}m}%
	\newcolumntype{G}{>{\centering}m}%
	\newcolumntype{J}{>{\raggedright\let\newline\\\arraybackslash}X}%
	\renewcommand{\arraystretch}{1.15}%
	\begin{tabularx}{\textwidth}{Z{\firstcolwidth+.2cm}lG{2.7cm}J}
		\toprule[0.1ex]
		Test case & & Time in seconds & Detected defect \\
		\midrule[0.3ex]
		\multirow{6}{\firstcolwidth}{Treiber's stack, bad linearisation point}
			& \texttt{@1}, early	& 0.05	& value loss                           	\\
			& \texttt{@1}, late 	& 0.07	& value out of thin air                	\\
			& \texttt{@2}, early	& 0.08	& multiple linearisation events emitted	\\
			& \texttt{@2}, late 	& 0.05	& value loss                           	\\
			& \texttt{@3}, early	& 0.02	& value duplication                    	\\
			& \texttt{@3}, late 	& 0.02	& value out of thin air                	\\
		\hdashline[1pt/1pt]
		Treiber's stack, swapped statements &
			& 0.001 & returned value stems from freed cell \\
		\hdashline[1pt/1pt]
		Treiber's stack, age fields discarded &
			& 0.001 & strong pointer race detected \\
		\hdashline[1pt/1pt]
		\multirow{6}{\firstcolwidth}{Michael\&Scott's queue (with false-positive prevention), bad linearisation point}
			& \texttt{@1}, early	& 170	& fifo property violated               	\\
			& \texttt{@1}, late 	& 2.7	& value out of thin air                	\\
			& \texttt{@2}, early	& 4.1	& multiple linearisation events emitted	\\
			& \texttt{@2}, late 	& 5.1	& value loss                           	\\
			& \texttt{@3}, early	& 0.2	& duplicate output                     	\\
			& \texttt{@3}, late 	& 0.4	& duplicate output                     	\\
		\hdashline[1pt/1pt]
		Swapped Statements, Michael\&Scott's queue &
			& 3.78 & returned value stems from freed cell \\
		\hdashline[1pt/1pt]
		Michael\&Scott's queue, age fields discarded &
			& 0.13 & strong pointer race detected \\
		\bottomrule[0.1ex]
	\end{tabularx}
\end{table}


\end{document}